\documentclass[leqno,onefignum,onetabnum]{siamart1116}
\usepackage[T1]{fontenc}
\usepackage[utf8]{inputenc}
\numberwithin{theorem}{section}
\numberwithin{equation}{section}

\usepackage{amssymb,amsmath,bm}
\usepackage{color,tikz}

\usepackage{latexsym,graphics,amscd}
\usepackage{graphicx,subfigure}
\usepackage{comment}
\usepackage{soul}

\usepackage{amsfonts}
\usepackage{mathrsfs}
\usepackage{rotating}
\usepackage{multirow}
\DeclareMathAlphabet{\pcal}{OMS}{zplm}{m}{n}

\newcommand{\cS}{{\cal S}}
\newcommand{\cA}{{\cal A}}

\newcommand{\cB}{{\cal B}}
\newcommand{\cP}{{\cal P}}

\newcommand\norm[1]{\lVert#1\rVert}

\newcommand\bone{\mathbf{1}}

\renewcommand{\bar}{\overline}

\def\xvec{{\bf x}} 
\def\yvec{{\bf y}}

\def\uvec{{\bf u}}
\def\vvec{{\bf v}}

\def\bvec{{\bf b}}

\def\qvec{{\bf q}}

\def\tvec{{\bf t}}

\newcommand{\bs}{\boldsymbol}
\newcommand{\R}{{\mathbb{R}}}

\newcommand{\calL}{{\mathcal{L}}}
\newcommand{\Rn}{\R^n}

\newcommand{\Rnn}{{\R^{n\times n}}}

\newcommand{\wt}{\widetilde}

\newsiamthm{rem}{Remark}
\newsiamthm{note}{Note}
\newsiamthm{conj}{Conjecture}
\newsiamthm{prop}{Proposition}
\newsiamthm{defn}{Definition}
\newsiamthm{example}{Example}

\newcommand{\vertiii}[1]{{\left\vert\kern-0.25ex\left\vert\kern-0.25ex\left\vert #1 
    \right\vert\kern-0.25ex\right\vert\kern-0.25ex\right\vert}}

\newcommand{\dref}[1]{Definition~\ref{#1}}

\newcommand{\sref}[1]{Section~\ref{#1}}

\newcommand{\Hmx}{\mathit \Theta}
\renewcommand{\epsilon}{\varepsilon}
\renewcommand{\tilde}{\widetilde}

\usepackage{soul}

\title{Node and layer eigenvector centralities for multiplex networks\thanks{Author's accepted version: this is the peer-reviewed version of this manuscript, which is now   published on SIAM Journal on Applied Mathematics \url{https://doi.org/10.1137/17M1137668}. \funding{The work of F.T. was funded by the European Union's Horizon 2020 research and innovation programme under the MarieSk\l odowska-Curie individual fellowship ``MAGNET'' grant agreement no.\ 744014. The work of F. A. has been supported by by the 
 EPSRC  grant EP/M00158X/1. The work of A. G. has been supported  by the ERC grant 307793 ``NOLEPRO''.}}}
\author{Francesco Tudisco\footnotemark[2], Francesca Arrigo\footnotemark[2] and Antoine Gautier\footnotemark[3]}
\begin{document}
\maketitle

 \renewcommand{\thefootnote}{\fnsymbol{footnote}}
  \footnotetext[2]{Department of Mathematics and Statistics, 
 University of Strathclyde, Glasgow, G11XH, UK (\email{f.tudisco@strath.ac.uk}, \email{francesca.arrigo@strath.ac.uk}).} 
 \footnotetext[3]{Department of Mathematics and Computer Science, 
 Saarland University, Saarbruecken, 66123, Germany  (\email{antoine.gautier@uni-saarland.de}).}

 \renewcommand{\thefootnote}{\arabic{footnote}}

\begin{abstract}
Eigenvector-based centrality measures are among the most popular centrality measures in network science. 
The underlying idea is intuitive and the mathematical description is extremely simple in the framework of standard, mono-layer networks. Moreover, several efficient computational tools are available for their computation.  
Moving up in dimensionality, several efforts have been made in the past to describe an eigenvector-based centrality measure that generalizes Bonacich index to the case of multiplex networks. 
In this work, we propose a new definition of eigenvector centrality that relies on the Perron eigenvector of a multi-homogeneous map defined in terms of the tensor describing the network. We prove that existence and uniqueness of such centrality are guaranteed under very mild assumptions on the multiplex network. 
Extensive numerical studies are proposed to test the newly introduced centrality measure and to compare it to other existing eigenvector-based centralities. 

\end{abstract}



\pagestyle{myheadings}
\thispagestyle{plain}
\markboth{{\sc F.~Tudisco, F.~Arrigo, and A.~Gautier}}{Centrality measures for multiplex networks}

\section{Introduction}

One of the main goals of network analysts and data scientists is to 
identify  relevant components in a network.  
Networks are often represented as matrices; hence,  tools from matrix analysis prove to be useful in addressing various problems, such as the detection of communities and anti-communities  \cite{newman2006finding,fasino2014algebraic,traag2011narrow,fasino2016generalized,fasino2017anticomm}, the partition of the network  into clusters \cite{ng2001spectral,von2007tutorial,chiang2014prediction,mercado2016clustering}, or the identification of central nodes, edges or paths  \cite{estrada2010network,langville2012s, borgatti2006graph, GHN17, AGHN17a}. 
Over the past years several authors, focusing on this latter problem and, in particular, on the problem of  identifying the most central nodes, have worked towards the definition of different {\it centrality measures}. 
A centrality measure is a real valued function of the set of nodes that is invariant under relabelling of the nodes and can thus be used to rank  them according to their importance. 
Measures based on eigenvectors or singular vectors of suitable  matrices,  henceforth referred to as  \textit{(linear) eigenvector centralities}, are among the most popular centrality measures. Relevant examples include, for instance, the Bonacich index \cite{B87} and the PageRank score \cite{brin1998anatomy}; this latter is  world-widely known for being employed by the  Google search engine.  The Perron-Frobenius theory for matrices is the key mathematical tool behind the success of linear eigenvector centralities; this theory  provides some  easily verified sufficient assumptions on the  network  that guarantee {\it (i)} existence and uniqueness of the centrality measure, and {\it (ii)} convergence of the  method used for its computation.

Recent years have seen a growing interest towards certain higher-order data structures,  usually referred to as {\it multilayer networks}, that are used to model networks where the entities interact in different ways. A particular instance of these networks is the one of \textit{multiplex networks}. A multiplex is a  collection  of graphs that share a common vertex set and  have (possibly)  different edge sets, each of which is  modelling a different form of interaction. Multiplexes arise in many contexts, such as  social networks, where individuals' relations can be tracked through different means (phone, Internet, in person, etc.), or transportation networks, where different transport means can be used to move across geographical destinations (car, train, plane, etc.).  

Several eigenvector-based centrality measures have been defined for muliplexes  in the last few years \cite{BNL14, SRC13, DSOGA13, DSOGA15}. A first contribution of our work is in reviewing the corresponding literature and gathering together these centrality scores.   
As for the matrix representation of a mono-layer network, the higher-order nature of multiple layers naturally suggests the use of tensors with more than two modes to describe the networks. 
However, in order to exploit the Perron-Frobenius theory for matrices for describing eigenvector-based centrality measures in this framework, a popular  idea in the current literature is to work with suitable matrix eigenvector equations. This approach is a form of ``linearization'' of the higher-order structure which, as we shall further  discuss later in the paper, might lead to  a loss of information. 

Our main contribution is the introduction of a novel centrality measure based on the Perron eigenvector of a multi-homogeneous map. These maps generalize the concept of homogenous functions and their definition is recalled at the beginning of Section \ref{sec:multi-hom} below. The particular mapping we will use is 
not linear nor multi-linear and   allows us to simultaneously derive a ranking for both nodes and layers in the multiplex. 
Our model follows a somewhat natural extension of the Bonacich index for mono-layer networks to the multilayer case. By exploiting the Perron-Frobenius Theorem for multi-homogeneous maps \cite{GTH17} we show that the  introduced  non-linearity is needed  in order to ensure existence and uniqueness of the eigenvector and, consequently,   for the well-posedness of the  model. 
Furthermore, existence and uniqueness are guaranteed for multiplex networks whose topology satisfies very mild conditions. In particular, connectedness of the multiplex network (or of its  individual layers) is no longer required. 
This is extremely  important as real-world data sets are typically very sparse and often not connected. Thus,  in most situations, other centrality measures fail to provide a well-defined centrality score.

We subdivide the discussion as follows: in the Section \ref{sec:back}  we briefly recall some background terminology and results. In Section \ref{sec:related} we review the eigenvector-based centrality scores that we found in the literature, highlighting their potential advantages and disadvantages. Then, in Section \ref{sec:main}, we describe  our  model and state and prove our   result on existence and uniqueness of the proposed centrality measure. Moreover, we describe  a power method iteration  which is proved to   converge to the newly introduced  centrality.  We also show an explicit a priori bound for the approximation error of the corresponding iteration.
Finally, in Section \ref{sec:ne} we propose several numerical experiments on two different multiplex networks, showing  evidence of the effectiveness of the model.

\section{Background}
\label{sec:back}
In this section we review some definitions and notations associated with graphs and multilayer graphs 
that will be used throughout the paper. 

We will denote by $\Rn_{\geq}$ the set of non-negative vectors of length $n$ and we will write equivalently 
that $\xvec\in\Rn_\geq$ or $\xvec\geq 0$. 
The set of vectors of length $n$ with positive entries 
$\xvec>0$ will be denoted by $\Rn_>$.  We will write $\xvec \geq \yvec$ (resp.\ $\xvec >\yvec$) if $\xvec - \yvec \geq 0$ (resp.\ $\xvec - \yvec >0$).

The symbol  $I$ will denote the identity matrix and  $\bone$ will denote the vector of all ones. 
The size of these objects will be clear from the context. 
However, when needed, we will make the dimension explicit by writing $I_n$ and $\bone_n$.  

\subsection*{Networks and centrality measures}

A {\it graph} $G = (V_n,E)$ consists of a pair of sets: a set $V_n$ of $n$ {\it nodes} 
and a set $E\subseteq V_n\times V_n$ of connections, or {\it edges}, between them. 
Every graph can be represented by means of a non-negative  matrix $A = (A_{ij}) \in \mathbb R^{n\times n}$, called the {\it adjacency matrix} of the graph, whose $(i,j)$th entry 
is the weight of the edge connecting node $i$ to node $j$, if present, and zero otherwise. 
A graph is said to be {\it undirected} if for all $(i,j)\in E$, then $(j,i)\in E$  and the two edges have the same weight. Equivalently, $G$ is undirected if its adjacency matrix is symmetric, {\it directed} otherwise. 
A network is said to be {\it unweighted} if all its edges have the same weight, which can thus be considered to be unitary, and {\it weighted} otherwise. 
In this paper, unless otherwise specified, we will consider weighted undirected  networks.

One of the main points addressed by researchers in network science is how to identify which are the ``most relevant'' nodes in a network. 
In order to quantify the importance of nodes, several {\it centrality measures} have been introduced over the years~\cite{BookE,BookEK,BookN}. 
A centrality measure is a real valued function of the set of nodes that is invariant under graph isomorphisms. 
One of the simplest centrality measure one can introduce is the {\it degree centrality}, which assigns to each node a score that 
corresponds to the number of its immediate neighbors; thus, for an unweighted network,  the degree centrality of node $i$ is the $i$th element of the vector $A\bone$.

Among the several centrality measures that are found in the literature, we want to focus here on those that 
rely on the use of eigenvectors of certain matrices associated with graphs. 
The {\it eigenvector centrality} was introduced by Bonacich in~\cite{B72a,B72b}. 
The underlying idea is that a node is more important if it is connected to nodes that are themselves important; therefore, the eigenvector centrality of a node $i$ is 
defined as
\begin{equation}\label{eq:standard_eig_centrality}
    \lambda x_i = \sum_{j=1}^n A_{ij}x_j\, ,
\end{equation}
for some $\lambda >0$.

This interpretation leads to the more formal definition of the eigenvector centrality of node $i$ as the $i$th component 
of the Perron vector of the adjacency matrix of the graph~\cite{BookE,BookHJ}. 
In order for this centrality vector to be well defined, i.e., for it to be positive and unique, the matrix $A$ has to satisfy certain properties. For example, a sufficient condition is its irreducibility, which is equivalent to assuming the strong connectedness of the underlying graph.\footnote{A graph is {\it strongly connected} if there exists a walk going from each node to any other node in the network.}


\subsection*{Multiplex networks}
A {\it multilayer network}~\cite{KAB14} is a more general system than a graph and can be used to model situations in which different types of 
interactions occur. 
In a multilayer graph, in addition to nodes and edges, {\it layers} are also present. 
Each node now belongs to a subset of the set of layers and interactions can occur through edges that 
exist within a layer or across layers. 
In the remainder of this work, we will consider a particular type of multilayer networks, namely {\it multiplex networks}~\cite{BBC14}. 
Multiplex networks, sometimes referred to as {\it color-graphs}~\cite{BNL14,BBC14,Bi13,MRM10,NBLB13}, are multilayer networks that are {\it node-aligned} and have 
{\it diagonal couplings}. More precisely, a multilayer is said to be node-aligned when all its layers contain the same set of nodes, and 
it is said to have diagonal couplings when the nodes are  identified across the different layers and there are no connections between nodes that belong to two different layers.      
Hence, a multiplex network can be represented as a collection of graphs:
\[
\mathcal G = \{G^{(\ell)} = (V_n,E^{(\ell)}) \}_{\ell\in V_L}
\]
where $V_n = \{1,2,\ldots,n\}$ is the set of nodes, $V_L=\{1,\dots,L\}$ is the set of layers and $E^{(\ell)}\subset V_n\times V_n$ is the set of edges 
on layer $\ell$. 
For every $\ell\in V_L$, the graph $G^{(\ell)}$ is associated with a non-negative adjacency matrix $A^{(\ell)} = (A_{ij}^{(\ell)})\in \mathbb R^{n\times n}$. 
Thus, the multiplex network can be represented by 
means of a $3^{\rm rd}$-order tensor $\cA = (\cA_{ij\ell})$, called the {\it adjacency tensor}, whose entries are
\[
\cA_{ij\ell} = A_{ij}^{(\ell)} = 
\left\{
\begin{array}{ll}
w_\ell(i,j) & \text{if }(i,j)\in E^{(\ell)} \\
0 & \text{otherwise}
\end{array}
\right.
\]
where $w_\ell(i,j)$ is a positive number representing the strength of the  connection between  node  $i$  and  node $j$ in layer $\ell\in V_L$. 
In this work, we will focus on the case of undirected weighted multiplex networks, i.e., multiplexes whose layers contain undirected networks.

Fairly often, multiplexes have been represented by means of ``vectors of the adjacency 
matrices of the $L$  layers'' (see, e.g., \cite{BNL14,NBLB13}), which 
correspond to the mode-1 unfolding of the original $3^{\rm rd}$-order tensors~\cite{KB09}   
\[\cA_{(1)} = [A^{(1)}, A^{(2)},\ldots, A^{(L)}],\] 
or using $4^{\rm th}$-order tensors. 
This latter description can be found in~\cite{DDSCK13,KAB14}, where the authors 
introduced the {\it multilayer adjacency tensor} 
$\cB = (\cB_{ij\ell\kappa})\in\R^{n\times n\times L\times L}$. 
Its entries $\cB_{ij\ell\kappa}$, for $i,j\in V_n$ and $\ell,\kappa\in V_L$,  
are positive  if there is an edge going from node $i$ in layer $\ell$ to node $j$ in layer $\kappa$. 
It is readly verified that $\cB_{::\ell\ell} = A^{(\ell)}\in\Rnn$ is the adjacency matrix of the graph appearing in layer $\ell \in V_L$; these matrices are usally referred to as {\it intra-layer adjacency tensors}. 
Similarly, the matrices $\cB_{::\ell\kappa} = D^{(\ell,\kappa)}\in\Rnn$ for $\ell,\kappa \in V_L$, $\ell \neq \kappa$, usually referred to as 
the {\it inter-layer adjacency tensors}, contain information about the connections between 
any two nodes when one is lying on layer $\ell$ and the other on layer $\kappa$. 
It is immediate to verify that, in the case of multiplex networks, $D^{(\ell,\kappa)} = I$ for all $\ell,\kappa\in V_L$, $\ell\neq \kappa$.

\section{Related works}
\label{sec:related}
In the following we briefly overview the eigenvector-based centrality measures that
we found in the literature. 
This list is exhaustive, to the best of our knowledge. 
For the sake of precision, let us point out that in this section we will often use the term ``centrality'' with a little abuse of notation; indeed, 
in the following we will often assign to each node a {\it vector} of scores, rather than just one value. 

\subsubsection*{Matrix-based centrality indices} 
In order to retrieve eigenvector centrality measures in the setting of multiplex networks, 
the authors of~\cite{BNL14} 
propose to compute the eigenvector centrality $\qvec_{\ell}$ of $G^{(\ell)}$ for all $\ell\in V_L$ as defined in \eqref{eq:standard_eig_centrality} and then to use as a measure of centrality for a node $i\in V_n$ 
the $i$th row of the matrix 
$Q = [\qvec_1,\qvec_2,\ldots,\qvec_L]\in\R^{n\times L}$.  
Note that, for this measure to be well defined, each graph $G^{(\ell)}$ has to fulfill suitable assumptions, such as being strongly connected.

Similarly, the authors of~\cite{SRC13} defined matrices that, in the spirit of $Q$, encode in their rows the information about the importance of the corresponding node in each layer. 
The {\it local heterogeneous eigenvector-like centrality} of  $G^{(\ell)}$, denoted by $\qvec_{\ell}^\star$, is defined 
as the eigenvector associated to the leading eigenvalue of the matrix 
$A_\star^{(\ell)} = \sum_{\kappa=1}^Lw_{\ell\kappa}A^{(\kappa)}$, 
where $W = (w_{\ell\kappa})\in \R^{L\times L}$ is a non-negative matrix 
called the {\it influence matrix}. 
Each of its entries $w_{\ell\kappa}$ accounts for the influence of layer $\kappa$ over layer $\ell$.
Then, the {\it local heterogeneous eigenvector-like centrality matrix} of the multiplex network is the matrix 
whose $\ell$th column is the vector $\qvec_{\ell}^\star$. Note that, since  $A_\star^{(\ell)}$ is obtained through a form of weighted average of the adjacency matrices of the layers,  positivity and uniqueness of the centrality $\qvec_{\ell}^\star$ can be guaranteed without necessarily requiring strong connectedness of each $G^{(\ell)}$. In fact, depending on the influence matrix, each $A_\star^{(\ell)}$ can be irreducible even though none of the $A^{(\ell)}$ is.

Finally, the {\it global heterogeneours eigenvector-like centrality matrix} 
of the multiplex is the reshaping into a $n\times L$ matrix of the eigenvector associated with the leading eigenvalue of 
the $nL\times nL$ matrix obtained by performing the 
Khatri-Rao product  \cite{KR68} of the influence matrix $W$ and the 
mode-1 unfolding of $\cA$~\cite{SRC13}
\begin{equation}\label{eq:global-heterogeneous-matrix}
 W\ast\cA_{(1)} = 
\left[
\begin{array}{cccc}
w_{11}A^{(1)} & w_{12}A^{(2)} & \cdots & w_{1L}A^{(L)}\\
w_{21}A^{(1)} & w_{22}A^{(2)} & \cdots & w_{2L}A^{(L)} \\
\vdots & \vdots & \ddots & \vdots \\
w_{L1}A^{(1)} & w_{L2}A^{(2)} & \cdots & w_{LL}A^{(L)}
\end{array}
\right]. 
\end{equation}

All the centrality measures listed so far are described as $n\times L$ non-negative matrices. 
To retrieve a proper centrality score, i.e., a non-negative real number, one has to {\it aggregate} the values appearing in the centrality vectors according to some heuristics. 
In~\cite{BNL14} the authors propose to use weighted sums, so that the centrality of a node will be the corresponding entry of the vector $Q{\bs\omega}$, where ${\bs\omega}\in\R_>^L$ is a vector of weights. 
Different choices of ${\bs\omega}$ provide different rankings. 
If no intuition/information is available concerning the importance of layers, then  the (possibly)  most appropriate choice is to pick ${\bs\omega} =\bone$, thus deriving 
\begin{equation}
{\tt eig\_cen}(i) = (Q\bone)_i = \sum_{\ell=1}^L (\qvec_\ell)_i, \quad i=1,\dots,n\,. 
\label{eq:EC}
\end{equation}

Should information concerning the importance of the layers in the network be available, 
a more inferred choice of the weights may be performed: the more important a layer, the higher the value of the corresponding weight.

The authors of~\cite{BNL14} described yet another way to overcome the issue of having one vector in $\R^L_\geq$ to describe the centrality of each node. 
Instead of aggregating the data {\it after} having computed the importance of the nodes in each layer, they 
first aggregate the layers, building the matrix ${A_{\mathrm{agg}}}({\bs \omega}) = \sum_{\ell=1}^L\omega_{\ell}A^{(\ell)}$
 for some vector ${\bs \omega} = (\omega_{\ell}) > 0$, and
then they use the entries of the eigenvector associated with the leading eigenvalue 
of this matrix  as centrality scores for the nodes. 
As in the previous  centrality measures, 
the choice of the vector of weights influences the resulting ranking of nodes, as one would expect. 
When ${\bs\omega} = \bone$, we recover the {\it uniform eigenvector-like centrality vector}, independently introduced in \cite{SRC13}, 
which is entry-wise defined as 
\begin{equation}\label{eq:agg_eig}
 {\tt agg\_eig}(i) = u_i, \quad i=1,\dots,n\, ,
\end{equation}
where $\uvec$ is the eigenvector associated to 
the leading eigenvalue of the {\it aggregate adjacency matrix} $A_{\mathrm{agg}}(\bone) = \sum_{\ell} A^{(\ell)}$.

Before moving on to the description of the last eigenvector-based centrality measure, let us briefly comment on the relationships among the indices listed so far, and their potential drawbacks.

Firstly, the methods that require to work independently on each of the layers are arguably not truly exploiting the multilayer structure of 
the multiplex: they are rather just considering  it as  a set of graphs which do not  share  anything except for  the size of the set of nodes. 
However, a multiplex is a much richer structure, since identification between the nodes also implies, e.g., that certain connections that, say, appear in more 
than one layer, or lie in a more important layer, are stronger than others. 
This aspect is completely overlooked by measures like the one deduced from $Q$, for example, unless one has some {\it a priori} information that 
cannot be deduced from the single layers and that can thus be used to aggregate the data in a more informed way.   

Similarly, those methods that first {\it aggregate} the matrices and then compute the eigenvector centrality of the derived matrix ${A_{\mathrm{agg}}}({\bs \omega})$ 
are disregarding some information. Indeed, not all the layers might be equally important: some might be more relevant than others, and merging everything 
together leads to the loss of this aspect. 
 
 Finally, let us point out that the local/global heterogeneous eigenvector-like centrality measures require the knowledge of how each layer influences all the others. This form of a priori knowledge is often not available in practice and, in this case, there are two standard and somewhat natural choices that can be made  for the influence matrix: either $W = I$ or $W = \bone\bone^T$. 
 
The first choice  corresponds to a mutiplex in which every layer is just influencing itself, and hence each layer exists independently from 
the others. This follows easily from the fact that, if $W = I$, then $A_\star^{(\ell)} = A^{(\ell)}$ for all $\ell\in V_L$. 
Hence, with this choice of $W$, the local heterogeneous eigenvector-like centrality matrix reduces to $Q$ and the associated node centrality boils down to \eqref{eq:EC}; instead, when addressing the  global heterogeneous eigenvector-like centrality several issues may arise. For example, the matrix \eqref{eq:global-heterogeneous-matrix}  $I\ast\cA_{(1)} = {\rm diag}(A^{(1)},A^{(2)},\ldots,A^{(L)})$ is now reducible, regardless of the edge pattern of  the multiplex. This implies that the corresponding Perron eigenvector, if uniquely determined, may fail to be positive. 
More precisely, assume that the adjacency matrices of $m$ layers, say $A^{(1)}, \dots, A^{(m)}$ for simplicity, have the same spectral radius as $I\ast\cA_{(1)}$ with corresponding Perron eigenvectors $\qvec_i$ for $i =1,\ldots,m$. Then any linear combination of $\wt{\qvec}_1, \dots, \wt{\qvec}_m$ is a Perron eigenvector of $I\ast\cA_{(1)}$, with 
$$
\wt\qvec_i = 
\begin{bmatrix}
0 & \cdots & 0 & \qvec_i^T & 0 & \cdots & 0
\end{bmatrix}^T
$$
and where, for any $i = 1,2,\ldots,m$, the nonzero entries in $\wt\qvec_i$ are in the positions $(i-1)n+1, \dots, in$.


On the other hand, if $W = \bone \bone^T$, each layer  equally influences all the others in the network, itself included. 
In this framework, $ A_\star^{(\ell)} = {A_{\mathrm{agg}}}(\bone)$ for all $\ell$ and thus the local heterogeneous eigenvector-like centrality matrix 
is $\uvec\,\bone_L^T\in\R^{n\times L}$, where $\uvec$ is the uniform eigenvector-like centrality vector \eqref{eq:agg_eig}. 
Similarly, it is not difficult to prove  that $\bone_L\otimes\uvec\in\R^{nL}$ is the eigenvector associated to the dominant eigenvalue of $\bone\bone^T\ast\cA_{(1)}$, where $\otimes$ is 
the Kroneker product~\cite{BookHJ}. 
Thus, if all the layers have the same influence on all the others, the local and global heterogeneous eigenvector-like  centralities reduce to the 
uniform eigenvector-like centrality \eqref{eq:agg_eig}.

\subsubsection*{A $4^{\rm th}$-order  tensor-based centrality index}
The last eigenvector-based centrality measure that we found in the 
literature is described in~\cite{DSOGA13, DSOGA15} and relies on the use of the multilayer adjacency tensor $\cB$. 
The authors use the matrix $F = (F_{i\ell})\in\R^{n\times L}$ defined via the equations
\begin{equation}\label{eq:numero}
\sum_{j=1}^n\sum_{\kappa = 1}^L\cB_{ij\ell\kappa}F_{j\kappa}= \lambda_1F_{i\ell}, 
\end{equation}
for $i\in V_n$ and for $\ell\in V_L$.
Then, the {\it eigenvector versatility} of a node $i$ is defined as the  $i$th row of the matrix $F$. 
As before, the proposed centrality for each node is described by means of a vector in $\R^L$. 
Thus, in order to have a more compact way of describing the centrality of the nodes, the authors 
proposed to use as the eigenvector versatility of node $i$ the quantity 
\[
{\tt eig\_ver}(i) = (F\bone)_i. 
\]
\begin{rem}
{\rm As for the matrix-based centralities discussed above, one may actually want to use a more general definition for the eigenvector versatility of nodes,  where  a vector of weights ${\bs \omega}>0$ is used insted of $\bone$ to aggregate the centralities over the layers. 
This vector encodes any available information regarding the relevance of layers that is not 
deducible from the structure of the network. }
\end{rem}
\begin{rem}{\rm 
The eigenvector versatility of a node is again a vector, corresponding to a row of an $n\times L$ matrix. 
Although this approach may appear to be  equivalent  to those described in the previous subsection, there is one
main difference. Indeed, in this case we are no longer considering each layer independently. 
In fact, we are using a representation of the multiplex as a $4^{\rm th}$-order tensor, whose 
description encodes the presence of links across layers that represent the diagonal couplings.  
However, in the case of multiplexes, the inter-layer links do not exist in practice and thus this representation forces the introduction of additional artificial data. 
}
\end{rem}

Following~\cite{DSOGA13}, in order to compute $F$, we build the {\it supra-adjacency matrix} associated with the multilayer adjacency tensor $\cB$, which is $nL\times nL$ block matrix of the form 
\[
B = \left[
\begin{array}{cccc}
A^{(1)}    & D^{(1,2)} & \cdots     & D^{(1,L)}     \\[3pt]
D^{(2,1)} & A^{(2)}    & \ddots     & \vdots     \\[3pt]
\vdots & \ddots & \ddots     & D^{(L-1,L)} \\[4pt]
D^{(L,1)} & \cdots & D^{(L,L-1)} & A^{(L)}
\end{array}
\right]
\]
where the $A^{(\ell)}$, for $\ell = 1,2,\ldots,L$, are the adjacency matrices of the graphs appearing on each layer, and $D^{(\ell,\kappa)}$, for $\ell,\kappa\in V_L$,
 are the inter-layer adjacency tensors.  
If we now let $\lambda$ be the spectral radius of $B\geq 0$, then $\lambda=\lambda_1$ and the associated eigenvector (if uniquely determined) is ${\tt vec}(F)$, where {\tt vec} is the standard vectorization operator, and $\lambda_1$ and $F$ are as in \eqref{eq:numero}.

In the case of multiplexes, the matrix $B$ reads
$$
B = \left[
\begin{array}{cccc}
A^{(1)}    & I      & \cdots     & I    \\[3pt]
I      & A^{(2)}    & \ddots     & \vdots    \\[3pt]
\vdots & \ddots & \ddots     & I \\[4pt]
I      & \cdots & I          & A^{(L)} 
\end{array}
\right], $$
or, equivalently
\begin{equation}
B = {\rm diag}(A^{(1)},A^{(2)},\ldots,A^{(L)}) + (\bone_L\bone_L^T - I_L)\otimes I_n.
\label{eq:supAdj2}
\end{equation}

Before moving on to the introduction of our new eigenvector centrality, let us briefly comment on the connectivity assumptions  required by the previous models. As already mentioned, in order for {\tt eig\_cen} to be well-defined, each layer of the undirected multiplex $\mathcal G$ has to be  connected. This is clearly a very strong assumption on the topology of the multiplex. 
Concerning the local,  global and  uniform heterogeneous eigenvector-like centralities, one needs the  strong connectedness of the merged graph $G_{\mathrm{agg}}=(V_n, \cup_\ell E^{(\ell)})$, whose set of edges is the union of the edges of the layers. The same property is required by the eigenvector versatility measure. In fact  it is easy to verify that the following property holds:
\begin{prop}\label{pro:versatility}
The matrix $B$ defined as in~\eqref{eq:supAdj2} is irreducible if and only if $G_{\mathrm{agg}}$ is strongly connected.
\end{prop}

The eigenvector versatility and the family of heterogeneous eigenvector-like centrality measures  require weaker conditions on the topology of the multiplex than those required by {\tt eig\_cen}. 
However, as we will see in the next section, the eigenvector centrality that we propose in this paper has even weaker requirements. 

To help intuition, we display in Figure~\ref{fig:example_connectedness} three different multiplex networks with four nodes and two layers. 
The centrality measure introduced in this paper is well-defined for all three multiplexes. 
However,  since the three networks have very different connectivity properties, this is no longer the case for the other centrality measures discussed so far. 
Specifically, all the layers in the leftmost network are connected and thus all the measures are well-defined. Concerning the multiplex in the center,  \texttt{eig\_cen} is not well defined since the aggregate network is connected whilst the individual layers are not. 
Finally, in the rightmost network, the layers nor the aggregate network are connected and thus 
\texttt{eig\_cen}, all the eigenvector-like centrality measures and the eigenvector versatility are not well-defined.

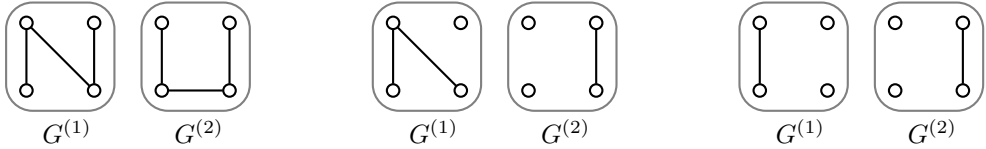
\begin{figure}
    \begin{tikzpicture}[scale=.9]
    \draw[gray,thick,rounded corners=10pt] (-.3,-.3) rectangle (1.3,1.3);
    \node at (.6,-.6) {$G^{(1)}$};
    \node[scale=.5,circle,draw=black,thick] (a1) at (0,0) {};
    \node[scale=.5,circle,draw=black,thick] (a2) at (0,1) {};
    \node[scale=.5,circle,draw=black,thick] (a3) at (1,1) {};
    \node[scale=.5,circle,draw=black,thick] (a4) at (1,0) {};
    \path[thick] (a1)edge(a2) (a2)edge(a4) (a4)edge(a3); 

    \draw[gray,thick,rounded corners=10pt] (1.7,-.3) rectangle (3.3,1.3);
    \node at (2.55,-.6) {$G^{(2)}$};
    \node[scale=.5,circle,draw=black,thick] (b1) at (2,0) {};
    \node[scale=.5,circle,draw=black,thick] (b2) at (2,1) {};
    \node[scale=.5,circle,draw=black,thick] (b3) at (3,1) {};
    \node[scale=.5,circle,draw=black,thick] (b4) at (3,0) {};
    \path[thick] (b1)edge(b2) (b1)edge(b4) (b4)edge(b3);
    \end{tikzpicture}
    \hfill
    \begin{tikzpicture}[scale=.9]
    \draw[gray,thick,rounded corners=10pt] (-.3,-.3) rectangle (1.3,1.3);
    \node at (.6,-.6) {$G^{(1)}$};
    \node[scale=.5,circle,draw=black,thick] (a1) at (0,0) {};
    \node[scale=.5,circle,draw=black,thick] (a2) at (0,1) {};
    \node[scale=.5,circle,draw=black,thick] (a3) at (1,1) {};
    \node[scale=.5,circle,draw=black,thick] (a4) at (1,0) {};
    \path[thick] (a1)edge(a2) (a2)edge(a4);

    \draw[gray,thick,rounded corners=10pt] (1.7,-.3) rectangle (3.3,1.3);
    \node at (2.55,-.6) {$G^{(2)}$};
    \node[scale=.5,circle,draw=black,thick] (b1) at (2,0) {};
    \node[scale=.5,circle,draw=black,thick] (b2) at (2,1) {};
    \node[scale=.5,circle,draw=black,thick] (b3) at (3,1) {};
    \node[scale=.5,circle,draw=black,thick] (b4) at (3,0) {};
    \path[thick] (b4)edge(b3);
    \end{tikzpicture}
    \hfill
    \begin{tikzpicture}[scale=.9]
    \draw[gray,thick,rounded corners=10pt] (-.3,-.3) rectangle (1.3,1.3);
    \node at (.6,-.6) {$G^{(1)}$};
    \node[scale=.5,circle,draw=black,thick] (a1) at (0,0) {};
    \node[scale=.5,circle,draw=black,thick] (a2) at (0,1) {};
    \node[scale=.5,circle,draw=black,thick] (a3) at (1,1) {};
    \node[scale=.5,circle,draw=black,thick] (a4) at (1,0) {};
    \path[thick] (a1)edge(a2); 

    \draw[gray,thick,rounded corners=10pt] (1.7,-.3) rectangle (3.3,1.3);
    \node at (2.55,-.6) {$G^{(2)}$};
    \node[scale=.5,circle,draw=black,thick] (b1) at (2,0) {};
    \node[scale=.5,circle,draw=black,thick] (b2) at (2,1) {};
    \node[scale=.5,circle,draw=black,thick] (b3) at (3,1) {};
    \node[scale=.5,circle,draw=black,thick] (b4) at (3,0) {};
    \path[thick] (b4)edge(b3);
    \end{tikzpicture}
    \caption{Representation of three multiplex networks with $n = 4$ and $L=2$ and different connectivity properties. The leftmost network has all connected layers (and thus the aggregate network is connected too). The multiplex in the center is such that the aggregate network is connected whilst the layers are not; in the rightmost network, the layers nor the aggregate network are connected. }\label{fig:example_connectedness}
\end{figure}

\section{Main results}
\label{sec:main}

In this section we  introduce a model that will lead to the definition of a novel centrality measure based on the  $3^{\rm rd}$-order tensor representation of  multiplex networks. The new centrality measure is defined  for  nodes as well as for  layers; moreover, to be computed, it does not require  any further  knowledge of the network apart from its topology. Therefore, the importance of the layers will 
now be computed rather than just  inferred  from some non-topological information on the network. The model is built on   the non-negative Perron eigenvector of a nonlinear multi-homogeneous order-preserving map $f$ defined in terms of the $3^{\rm rd}$-order  tensor representing the network. 
We will use the entries of this non-negative vector to assign a score to each node and each layer  in the multiplex, 
thus defining what we will  call the {\it $f$-eigenvector centrality} for  multiplexes. 

As we have already seen in Section~\ref{sec:back}, the idea behind the eigenvector centrality of nodes in a standard, monolayer complex network is that 
the importance of a node has to be proportional to the importance of its 
neighbours~\cite{B72a,B72b, B87}.  
Here we want to exploit the same idea and to generalize it to a measure of centrality for nodes and 
layers in multiplex networks with undirected layers.

There are three main differences between the results we are about to present and those found in the literature and described in
Section~\ref{sec:related}. 
First, we will work directly on the $3^{\rm rd}$-order adjacency tensor and we will return centrality scores without the need of aggregating or unfolding data.  
Second, as already pointed out, our method provides both a centrality vector for the nodes and a centrality vector for the layers in the multiplex. 
None of the methods found in the literature addressed the problem of computing  centrality scores for layers. 
Third, the proposed centrality measures are uniquely defined regardless of the irreducibly pattern of the network.

\subsection{$f$-eigenvector centrality indices}\label{sec:multi-hom}
By pursuing a somewhat natural extension of node and edge eigenvector  centralities for mono-layer networks~\cite{AB16a,B72a,B72b}, 
one  may ideally want  to define the eigenvector centrality $x_i$ of node $i$ in such a way that it 
is proportional to the importance of its neighbouring nodes and to the importance $t_\ell$ of the layers 
on which the nodes are connected. 
Similarly, one could argue that the  
eigenvector centrality $t_\ell$ of  layer $\ell$, should  be proportional to the importance of the nodes connected through a link on such  layer; therefore  the quantities $x_ix_j$ can be considered as  a measure of the relevance of the link  between nodes 
$i$ and $j$, generalizing the definition of edge eigenvector centrality~\cite{AB16a}. 
This latter measure  is a global index that accounts for the total strength of the tie between nodes $i$ and $j$, without any  specification about the layers on which the interaction occurs. 

Although being very natural and being the basis for some very successful centrality scores for monolayer graphs, these linear relations are often not sufficient to describe a model that is 
mathematically well-defined. 
For this reason we slightly modify this intuitive definition by adding a nonlinear 
term. 
This will then lead to the definition of $x_i$ and $t_\ell$ as the entries of  the unique entry-wise non-negative eigenvector of a multi-homogeneous map~\cite{GTH17}. 
For the sake of completeness let us recall that $f=(f_1, \ldots,f_d):\R^{n_1}\times\dots\times\R^{n_d}\to \R^{n_1}\times\dots\times\R^{n_d}$ 
is said to be {\it multi-homogeneous} if there exist non-negative coefficients $\Hmx_{ij}$ such that 
\[
f_i(\xvec_1,\ldots,\lambda\xvec_j,\ldots,\xvec_d)=\lambda^{\Hmx_{ij}}f_i(\xvec_1,\ldots,\xvec_d)\, , 
\]
for any $\lambda\in \R$ and any $i,j=1,\dots,d$.  The matrix $\Hmx = (\Hmx_{ij})\in\R^{d\times d}$ is the so-called {\it homogeneity matrix} of $f$. 
It is easily verified that any square matrix is an example of multi-homogeneous map with $\Hmx=1$. 
A nonzero vector $(\xvec_1, \ldots,\xvec_d)$  is an {\it eigenvector}
for $f$ if there exist $\lambda_1, \ldots, \lambda_d\in\R$ 
such that $f_i(\xvec_1, \ldots,\xvec_d)=\lambda_i\xvec_i$, for any $i=1,2, \ldots, d$. 

Let $\cA = (\cA_{ij\ell})$ be the non-negative adjacency tensor of a possibly weighted multiplex  and let  $\alpha, \beta >0$.   
The proposed model can be formalized as the solution to the following system of nonlinear equations
\begin{equation}\left\{
\begin{array}{l}
 \sum_{j=1}^n\sum_{\ell=1}^L\cA_{ij\ell}\, x_j\, t_{\ell} = \mu\,  (x_i)^\alpha\\[5pt]
\sum_{i=1}^n\sum_{j=1}^n \cA_{ij\ell}\, x_i\, x_j = \lambda\,  (t_{\ell})^\beta\, 
\end{array} 
\right.
\label{eq:model}
\end{equation}
that has to be fulfilled for some positive scalars $\mu$ and $\lambda$. 
 Before moving on to our main result,  we observe  that if there exists $i$ such that $\cA_{ij\ell}=0$ for all $j,\ell$, then $x_i=0$. Similarly, if there exists $\ell\in V_L$ for which $\cA_{ij\ell}=0$ for all $i,j$, then $t_\ell=0$. It is thus clear that, depending on the nonzero pattern of $\cA$, any solution $(\xvec, \tvec)$ of \eqref{eq:model} is required to be zero on a  certain set of entries. This has a natural interpretation in terms of the multiplex, as for instance $\cA_{ij\ell}=0$ for all $i,j$ implies that no edges lie on layer $\ell$ and thus the importance $t_\ell$ of this layer is negligible.
Moreover, since we aim at defining centrality scores, we are interested in computing non-negative vectors. Let us thus introduce the following  set of pairs of vectors  \eqref{eq:model} 
\[
\cP_\cA^{\, n\times L} =\left\{(\xvec,\tvec)\in\R^n_\geq\times \R^L_\geq :%
\begin{array}{c}
 x_i \sim  \textstyle{\sum_{j,\ell}\cA_{ij\ell}}, \quad \forall i = 1, \dots, n\\
 t_\ell \sim \textstyle{\sum_{i,j}\cA_{ij\ell}}, \quad \forall \ell =1,\dots, L
\end{array}\right\}\, ,
\]
where, for any two numbers $x,y\geq 0$,  we write $x \sim y$ if there exist $C>0$ such that $x =  Cy$. The set $\cP_\cA^{\, n\times L}$ contains the non-trivial solutions to \eqref{eq:model} we are interested in. To help intuition consider the example case of a mono-layer network, where
$\cA=A$ is a matrix. Then  $\cP_A^{\, n\times 1}$ is the set of pairs $(\xvec, t)$ where $\xvec$ is any non-negative vector with the same nonzero pattern as the vector of degrees $A\mathbf 1$ and $t$ is any positive number, if $A$ is not the zero matrix, whereas $t=0$ otherwise.

Similarly to the matrix case, uniqueness of the solution to \eqref{eq:model} can be ensured only up to scalar multiples. Indeed,  if $(\xvec, \tvec)$ is a solution to \eqref{eq:model}, then for any $a,b>0$, $(a\xvec, b\tvec)$ is also a solution, with possibly different positive scalars $\lambda$ and $\mu$. Therefore, in order to ensure uniqueness, we shall further ask that  $(\xvec,\tvec)\in\mathcal S_\cA^{\, n\times L}$, where
\[
\mathcal S_\cA^{\, n\times L} = \{(\xvec, \tvec)\in \cP_\cA^{\, n\times L} : \|\xvec\|_1=\|\tvec\|_1=1\}\, .
\]

Consider the mapping $f=(f_1, f_2): \Rn_{\geq}\times\R^{L}_{\geq}\to\Rn_{\geq}\times\R^{L}_{\geq}$ defined by  
\begin{equation}\label{eq:F}
 f_{1}(\xvec,\tvec)_i = \Big(\sum_{j=1}^n\sum_{\ell=1}^L\cA_{ij\ell}x_jt_{\ell}\Big)^{1/\alpha}, \qquad
 f_{2}(\xvec,\tvec)_\ell = \Big(\sum_{i=1}^n\sum_{j=1}^n\cA_{ij\ell}x_ix_j\Big)^{1/\beta}\, ,  
\end{equation}
and its normalized version $g:\mathcal{S}_\cA^{\, n\times L}\to\mathcal{S}_\cA^{\, n\times L}$ defined as
\begin{equation}\label{eq:G}
g(\xvec,\tvec) = \left(\frac{f_1(\xvec,\tvec)}{\norm{f_1(\xvec,\tvec)}_1},\frac{f_2(\xvec,\tvec)}{\norm{f_2(\xvec,\tvec)}_1}\right).
\end{equation}

The importance of this map  follows from  the following crucial observation: if $(\xvec^*,\tvec^*)$ is a solution to \eqref{eq:model}, then  $(\xvec^*,\tvec^*)$ is a fixed point of the mapping $g$, i.e., $g(\xvec^*,\tvec^*)=(\xvec^*,\tvec^*)$. Conversely, as we discuss below, any fixed point of $g$ solves \eqref{eq:model}.

\begin{theorem}\label{thm:main_undirected}
Let $\cA\in\R_{\geq}^{n\times n\times L}$ be a nonzero non-negative tensor and let $\alpha,\beta>0$ be such that $2/\beta<(\alpha-1)$. Then, the system of  nonlinear equations \eqref{eq:model} has a unique non-negative solution  $(\xvec^*,\tvec^*)\in \mathcal{S}_\cA^{\, n\times L}$.
 \end{theorem}
\begin{proof} Let $f$ and $g$ be as in \eqref{eq:F} and \eqref{eq:G} respectively. The proof is in two steps: we first show that $g$ is a strict contraction on $\mathcal S_\cA^{\, n\times L}$ with respect to a metric $\delta_{\bvec}$, and then we conclude by using the Banach fixed point Theorem (see, e.g., \cite[Theorem 3.1]{BookKK}). Let us recall that, for any non-empty set of indices $\mathcal I$, the Hilbert metric on $P_{\mathcal I}=\big\{ \xvec :  x_i>0 \text{ if } i \in \mathcal I \text{ and } x_i=0 \text{ otherwise}\big\}$  is given by \begin{equation}\label{eq:Hilbertmetric}
d_{\mathcal{I}}(\xvec,\uvec) = \ln \Big( \max_{i\in \mathcal I}\frac{x_i}{u_i} \max_{j \in \mathcal I}\frac{u_j}{x_j}\Big) \qquad \forall \xvec,\uvec \in P_{\mathcal I}.
\end{equation}
In particular, the set $(\{\xvec \in P_{\mathcal I}\colon \norm{\xvec}_1=1\},d_{\mathcal I})$ is a complete metric space (see, e.g., \cite[Proposition 2.5.4]{BookLN}). Consider now the index sets  $\mathcal I\subseteq \{1,\ldots,n\}$ and $\mathcal J \subseteq \{1,\ldots, L\}$  so  that $\mathcal P_\cA^{\, n\times L} = P_{\mathcal I}\times P_{\mathcal J}$. Note that, since $\mathcal A$ is assumed to be nonzero, we have $\mathcal I, \mathcal J \neq \emptyset$. Thus, for every $\bvec = (b_1, b_2)\in\R^2_>$, we have that $(\mathcal S_\cA^{\, n\times L},\delta_{\bvec})$ is a complete metric space, with  $\delta_{\bvec}$ being the weighted product Hilbert metric  defined by
\begin{equation*}
\delta_{\bvec}\big((\xvec,\tvec),(\uvec,\vvec)\big) := b_1\, d_{\mathcal{I}}(\xvec,\uvec) +b_2\, d_{\mathcal{J}}(\tvec,\vvec).
\end{equation*}

Let us now discuss how to choose the weights $\bvec$ so that $g$ is a contraction with respect to $\delta_{\bvec}$. 

First note that, for all $(\xvec,\tvec) \in \mathcal S_\cA^{\, n\times L}$ and $c,\tilde c>0$, it holds 
\begin{equation}\label{multihomomat}
f_j(c\,\xvec,\tilde c\,\tvec) = c^{\Hmx_{j1}}\, \tilde c^{\, \Hmx_{j2}} f_j(\xvec,\tvec), \qquad \text{for } j = 1,2 
\end{equation} 
where the homogeneity matrix $\Hmx\in\R^{2\times 2}$ is
\[
\Hmx =
\begin{bmatrix}
\alpha^{-1} & \alpha^{-1} \\ 2\beta^{-1} & 0
\end{bmatrix}.
\]

 Furthermore, note that  $f$ is order-preserving since $\mathcal A$ has non-negative entries, i.e., 
 for all $(\xvec,\tvec),(\uvec,\vvec) \in \mathcal S_\cA^{\, n\times L}$ with $(\xvec,\tvec)\leq(\uvec,\vvec)$ it holds
\begin{equation}\label{opf}
f(\xvec,\tvec)\leq f(\uvec,\vvec).
\end{equation}

Now, let $(\xvec,\tvec),(\uvec,\vvec) \in \mathcal S_\cA^{\, n\times L}$. Combining \eqref{opf} and \eqref{multihomomat} 
we get
\begin{equation}\label{eq1}
f_j(\xvec,\tvec) \leq \Big(\max_{i\in\mathcal I} \frac{x_i}{u_i}\Big)^{\Hmx_{j1}}\Big(\max_{i\in\mathcal J}\frac{t_i}{v_i}\Big)^{\Hmx_{j2}}f_j(\uvec,\vvec), \qquad \text{for } j=1,2.
\end{equation}
By exchanging  $(\xvec,\tvec)$ and $(\uvec,\vvec) $ in the above equation, we further obtain
\begin{equation}\label{eq2}
f_j(\uvec,\vvec) \leq \Big(\max_{i\in\mathcal I} \frac{u_i}{x_i}\Big)^{\Hmx_{j1}}\Big(\max_{i\in\mathcal J}\frac{v_i}{t_i}\Big)^{\Hmx_{j2}}f_j(\xvec,\tvec), \qquad \text{for } j=1,2.
\end{equation}
Equations~\eqref{eq1} and~\eqref{eq2} together imply that for any $\bvec\in\R^2_>$
\begin{equation*}
\delta_{\bvec}\big(f(\xvec,\tvec),f(\uvec,\vvec)\big)\leq \max\Big\{\frac{(\Hmx^T\bvec)_1}{b_1},\frac{(\Hmx^T\bvec)_2}{b_2}\Big\}\, \delta_{\bvec}\big((\xvec,\tvec),(\uvec,\vvec)\big).
\end{equation*}
 
 The Collatz--Wielandt formula (see, f.i., \cite[Corollary 8.1.31]{BookHJ}) states that the vector $\bvec = (b_1, b_2)>0$ that minimizes the maximum in the above formula is the positive eigenvector  of $\Hmx^T$. 
In particular, if  we let   
$\rho = \frac{\sqrt{8\alpha+\beta}+\sqrt{\beta}}{2\alpha\sqrt{\beta}},$ $b_1 = \alpha \rho$, and $b_2=1$, we have $\Hmx^T \bvec = \rho \bvec$ and thus $\rho$ is a Lipschitz constant of $f$ with respect to $\delta_{\bvec}$. 
The assumption $\beta>2/(\alpha-1)$ ensures that $\rho<1$; indeed,
\begin{equation}\label{eq:rho<1}
8\alpha+\beta = \frac{2}{\alpha-1}4\alpha(\alpha-1)+\beta < \beta(4\alpha(\alpha-1)+1) = \beta (2\alpha-1)^2.
\end{equation}
Finally, note that, since $d_{\mathcal I},d_{\mathcal J}$ are projective, for every $(\xvec,\tvec),(\uvec,\vvec)\in\mathcal P_\cA^{\, n\times L}$ it holds
\begin{equation}\label{eq:projHilbert}
\delta_{\bvec}\big((c\,\xvec,\tilde c\, \tvec),(a\,\uvec,\tilde a\, \vvec)\big)=\delta_{\bvec}\big((\xvec,\tvec),(\uvec, \vvec)\big) \qquad \forall c,\tilde c,a, \tilde a >0.
\end{equation}
It thus follows that 
\begin{equation}\label{eq:contractg}
\delta_{\bvec}\big(g(\xvec,\tvec),g(\uvec,\vvec)\big)=\delta_{\bvec}\big(f(\xvec,\tvec),f(\uvec,\vvec)\big)\leq \rho\, \delta_{\bvec}\big((\xvec,\tvec),(\uvec,\vvec)\big),
\end{equation}
and hence $g$ is a strict contraction on the complete metric space $(\mathcal S_\cA^{\, n\times L},\delta_{\bvec})$. 
Thus, $g$ has a unique fixed point $(\xvec^*,\tvec^*)\in \mathcal S_\cA^{\, n\times L}$.

We conclude the proof by noting that every solution to \eqref{eq:model} is a fixed point of $g$, and thus it is unique;  conversely,   every fixed point $(\xvec,\tvec)\in \mathcal S_\cA^{\, n\times L}$ of $g$ is a solution to \eqref{eq:model} with $\mu = \norm{f_1(\xvec,\tvec)}_1$ and $\lambda = \norm{f_2(\xvec,\tvec)}_1$,  thus implying existence.
\end{proof}

Theorem~\ref{thm:main_undirected} above shows that the condition $2/\beta < \alpha - 1$ on the positive parameters $\alpha$ and $\beta$ is needed to guarantee existence and uniqueness of a solution to~\eqref{eq:model}
in $\mathcal S_\cA^{\, n\times L}$.   
Clearly, the natural choice $\alpha = \beta = 1$ does not satisfy the above condition and may give rise to an ill-posed problem, as the following simple example shows.

\begin{example}Let $n=L =2$ and consider the tensor $\mathcal A = \mathcal B/25$, where  
$$
\left\{
\begin{array}{llll} 
\mathcal B_{1,1,1} = 6 & \mathcal B_{1,2,1} = 199/7 & \mathcal B_{2,1,1} = 16/7 & \mathcal B_{2,2,1} = 11 \\ 
\mathcal B_{1,1,2} = 61/7 & \mathcal B_{1,2,2} = 6 & \mathcal B_{2,1,2} = 29 & \mathcal B_{2,2,2} = 16/7
\end{array}\right. \, .
$$
Then both $\xvec = (2,1)/3$, $\tvec = (1,2)/3$ and $\tilde \xvec = (1,3)/4$,  $\tilde \tvec = (3,1)/4$ solve \eqref{eq:model}. 
\end{example}

The proof of Theorem \ref{thm:main_undirected} relies on techniques proposed in \cite{GTH17}; however, due to the special structure of \eqref{eq:model}, our proof is shorter and simpler. Moreover, we have less restrictive assumptions on the entries of the tensor. The conclusion of the proof follows from the Banach fixed point Theorem. This has two main advantages: first, it ensures uniqueness of the solution and second, it naturally induces an iterative method for its computation which we discuss in more detail in Theorem \ref{cor:convergence}.

A few further  comments on Theorem~\ref{thm:main_undirected}  are in order. 
Firstly, we want to stress that we are not requiring $\cA$ to be necessarily defined as in \sref{sec:back}. 
The definition we gave is essentially the generalization of the adjacency matrix to the case of 
$3^{\rm rd}$-order tensors; however, this theorem applies to any non-negative tensor of order $3$ that one may want to use to describe the multiplex network under study.
Secondly, let us note that if the considered tensor $\cA$ is such that 
\begin{itemize}
\item[$(i)$]for all $\ell\in V_L$ there exist $i,j\in V_n$ such that $\cA_{ij\ell}>0$,  
\item[$(ii)$] for all $i\in V_n$ there exist $\ell\in V_L$ and $j\in V_n$ such that $\cA_{ij\ell}>0$,
\end{itemize}
then the unique solution to \eqref{eq:model} is entry-wise {\it positive}. 
The above requirements correspond to very mild conditions on the topology of the multiplex network. 
Indeed, $(i)$ requires that there are no empty layers, whereas $(ii)$ coincides with the requirement that all nodes must have at least one connection in at least one layer, i.e., 
the aggregate degree of every node must be positive. So the existence of  a unique positive solution is ensured by our nonlinear model \eqref{eq:model} under sensibly weaker conditions than the ``standard'' irreducibility assumption on the adjacency matrix of mono-layer networks. 

Finally, we want to stress that these non-negative vectors can indeed be used as centrality vectors, as any solution to \eqref{eq:model} does not depend on the labeling of nodes and layers.  In fact, let $\sigma:V_n\to V_n$ and $\pi:V_L\to V_L$ be  two permutations and define the tensor $\tilde \cA$ with entries $\tilde \cA_{ij\ell}=\cA_{\sigma(i)\sigma(j)\pi(\ell)}$. Then the eigenvector $(\tilde \xvec, \tilde \tvec)\in \cS_{\tilde \cA}^{\,n\times L}$ of $f$ defined in terms of $\tilde \cA$ is such that $\tilde x_i = x_{\sigma(i)}$ and $\tilde t_\ell = t_{\pi(\ell)}$, where $(\xvec, \tvec)\in \cS_\cA^{\,n\times L}$ is the solution to \eqref{eq:model} associated with $\cA$.

We are now able to give a definition of node and layer centrality for multiplex networks.

\begin{defn}
 Let $\cA\in\R^{n\times n\times L}$ be a nonzero $3^{\rm rd}$-order non-negative tensor with undirected layers and  let $\alpha,\beta >0$ be such that $2/\beta<(\alpha-1)$. Define $f$ as the multi-homogeneous function~\eqref{eq:F}.  
For any $i\in V_n$ and $\ell\in V_L$, we define the {\rm $f$-node eigenvector centrality of node $i$} as $C_f(i)=x_i$
and the {\rm $f$-layer eigenvector centrality of layer $\ell$} as  $L_f(\ell)=t_\ell$, where $(\xvec,\tvec)$ is the unique non-negative eigenvector of $f$ in $\mathcal S_\cA^{\, n\times L}$.
\label{def:eig}
\end{defn}
As discussed above, the proof structure of Theorem \ref{thm:main_undirected} naturally induces an iterative method for the computation of node and layer centrality for multiplex networks. This method is described in the following Theorem. Furthermore, partially inspired from results in \cite{gautier2016globally, gautier2017tensor, tudisco2017krylov}, we derive explicit convergence rates for our method. In particular, we describe explicit bounds on the number of iterations   required in order to obtain a desired accuracy. 
\begin{theorem}\label{cor:convergence}
Let $\cA$, $\alpha$ and $\beta$ be as in the hypotheses of  Theorem~\ref{thm:main_undirected}. Furthermore, let $(\xvec^*,\tvec^*)$ be the unique solution to \eqref{eq:model} and let $g$ be defined as in \eqref{eq:G}. Given any $(\xvec^{(0)},\tvec^{(0)})\in \R_>^n\times \R_>^L$ 
consider the sequence
\begin{equation}\label{eq:power_sequence}
 (\xvec^{(k+1)},\tvec^{(k+1)})=g(\xvec^{(k)},\tvec^{(k)}), \qquad k = 0,1,2,\ldots
\end{equation}
Then, $$\lim_{k\to \infty} (\xvec^{(k)},\tvec^{(k)})=(\xvec^*,\tvec^*).$$ 
Furthermore, if $\rho=\frac{\sqrt{8\alpha+\beta}+\sqrt{\beta}}{2\alpha\sqrt{\beta}}$ and $(\xvec^{(0)},\tvec^{(0)})=(\bone_n,\bone_L)$, then $\forall k=1,2,\ldots$
\begin{equation*}
\norm{(\xvec^{(k)},\tvec^{(k)})-(\xvec^*,\tvec^*)}_{\infty} \leq \rho^k\Big[\alpha\rho\Big(\frac{\max_{i\in\mathcal I} x_i^*}{\min_{i'\in\mathcal I} x_{i'}^*}\Big)+\Big(\frac{\max_{\ell\in\mathcal J} t_\ell^*}{\min_{\ell'\in\mathcal J} t_{\ell'}^*}\Big)\Big] 
\end{equation*}
with  $\mathcal I=\{i\colon x^{(1)}_i>0\}$ and $\mathcal J =\{\ell \colon t_{\ell}^{(1)}>0\}$. 
Therefore,  for every $\epsilon>0$ and for any $k$ such that 
\begin{equation}\label{eq:formula_k}
 k \geq \frac{\ln((1-\rho)\epsilon)-\ln(C)}{\ln(\rho)}\, ,    
\end{equation}
 we have $\norm{(\xvec^{(k)},\tvec^{(k)})-(\xvec^*,\tvec^*)}_{\infty} \leq \epsilon$ 
where 
$$C=\rho\ln\bigg(\max_{i,i'\in\mathcal I} \frac{\sum_{j=1}^n\sum_{\ell=1}^L\mathcal A_{ij\ell}}{\sum_{j=1}^n\sum_{\ell=1}^L\mathcal A_{i'j\ell}}\bigg)+\frac{1}{\beta}\ln\bigg(\max_{\ell,\ell'\in\mathcal J} \frac{\sum_{i=1}^n\sum_{j=1}^n\mathcal A_{ij\ell}}{\sum_{i=1}^n\sum_{j=1}^n\mathcal A_{ij\ell'}}\bigg)\, .$$
\end{theorem}

\begin{proof}
Let $\delta_{\bvec}$ be the weighted Hilbert metric defined in the proof of Theorem~\ref{thm:main_undirected} with $b_1 = \alpha \rho$ and $b_2=1$. Note that $\rho< 1$ by \eqref{eq:rho<1}. 

Let $(\tilde \xvec^{(0)},\tilde \tvec^{(0)})=(\xvec^{(0)}/\norm{\xvec^{(0)}}_1,\tvec^{(0)}/\norm{\tvec^{(0)}}_1)\in \mathcal S_\cA^{\, n\times L}$. Then we have $g(\tilde \xvec^{(0)},\tilde \tvec^{(0)})\!\!=g(\xvec^{(0)},\tvec^{(0)})$ and, by \eqref{eq:projHilbert}, 
\begin{equation*}\delta_{\bvec}\big((\tilde \xvec^{(0)},\tilde \tvec^{(0)}),(\xvec,\tvec)\big) = \delta_{\bvec}\big((\xvec^{(0)}, \tvec^{(0)}),(\xvec,\tvec)\big) \qquad \forall (\xvec,\tvec)\in \mathcal P_\cA^{\, n\times L}. 
\end{equation*}
Therefore, using \eqref{eq:contractg} together with the Banach fixed point Theorem, we obtain
\begin{equation}\label{eq:bound1}
\delta_{\bvec}\big((\xvec^{(k)},\tvec^{(k)}),(\xvec^*,\tvec^*)\big) \leq \rho^k \delta_{\bvec}\big((\xvec^{(0)},\tvec^{(0)}),(\xvec^*,\tvec^*)\big) \qquad \forall k=1,2,\ldots
\end{equation}
and
\begin{equation}\label{eq:bound}
\delta_{\bvec}\big((\xvec^{(k)},\tvec^{(k)}),(\xvec^*,\tvec^*)\big) \leq \frac{\rho^k}{1-\rho} \delta_{\bvec}\big((\xvec^{(1)},\tvec^{(1)}),(\xvec^{(0)},\tvec^{(0)})\big) \qquad \forall k=1,2,\ldots
\end{equation}
We now use these inequalities to prove the convergence rates.

For a non-negative vector $\xvec\in\R^{m}_{\geq}\setminus\{{\bf 0}\}$, let $\mathcal I_{\xvec}=\{i\colon x_i>0\}$ and define the entries of $\bar \xvec = (\bar x_i)\in\R^m$ as $\bar x_i = \ln(x_i)$ if $i\in\mathcal I_{\xvec}$ and $\bar x_i = 0$ otherwise. 
Then, for every $\xvec,\uvec\in\R^m_{\geq}\setminus\{{\bf 0}\}$ with $\mathcal I_{\xvec}=\mathcal I_{\uvec}$ and $\norm{\xvec}_1=\norm{\uvec}_1=1$, we have $\max_{i\in\mathcal I_{\xvec}} x_i/v_i\geq 1$ and $\max_{i\in\mathcal I_{\xvec}} v_i/x_i\geq 1$ so that
\begin{equation*}
d_{\mathcal I_{\xvec}}(\xvec,\uvec)\geq\ln\left(\max\left\{\max_{i\in\mathcal I_{\xvec}} \frac {x_i}{u_i}, \max_{i\in\mathcal I_{\xvec}} \frac{u_i}{x_i}\right\}\right)=\ln\left(\max_{i\in\mathcal I_{\xvec}} e^{|\ln(x_i)-\ln(u_i)|}\right)=\|\bar\xvec-\bar \uvec\|_\infty\,
\end{equation*}
where $d_{\mathcal I_{\xvec}}$ is defined as in \eqref{eq:Hilbertmetric}. 
Since $|e^a-e^b|\leq |a-b|\max\{e^a,e^b\}$ for every $a,b>0$, we deduce that
\begin{equation*} \|\bar \xvec-\bar\uvec\|_\infty \geq \|\xvec-\uvec\|_\infty \left(\max_i(\max\{x_i,u_i\})\right)^{-1}\geq  \|\xvec-\uvec\|_\infty\, .\end{equation*}
Now, $(\xvec^{(k)},\tvec^{(k)})\in \mathcal S_\cA^{\, n\times L}$ for every $k=1,2,\ldots$, and thus, 
as $2/\beta<\alpha-1$, we get
\begin{equation*}
\alpha\,\rho = \frac{\sqrt{8\alpha+\beta}+\sqrt{\beta}}{2\sqrt{\beta}} > \frac{\sqrt{\tfrac{16}{\beta}+8+\beta}}{2\sqrt{\beta}}+\frac{1}{2}= \frac{1}{2}\Big(1+\sqrt\frac{{\tfrac{16}{\beta}+8+\beta}}{{\beta}}\Big)>1.
\end{equation*}
It follows that for every $k=1,2,\ldots$
\begin{align*} 
 \norm{(\xvec^{(k)},\tvec^{(k)})-(\xvec^*,\tvec^*)}_{\infty} &\leq \alpha\,\rho\,\norm{\xvec^{(k)}-\xvec^*}_{\infty}+\norm{\tvec^{(k)}-\tvec^*}_{\infty} \\
                                                             &\leq \delta_{\bvec}\big((\xvec^{(k)},\tvec^{(k)}),(\xvec^*,\tvec^*)\big)
\end{align*}
This, together with the identity
\begin{equation*}
\delta_{\bvec}\big((\xvec^{(0)},\tvec^{(0)}),(\xvec^*,\tvec^*)\big) = \Big[\alpha\rho\Big(\frac{\max_{i\in\mathcal I} x_i^*}{\min_{i'\in\mathcal I} x_{i'}^*}\Big)+\Big(\frac{\max_{\ell\in\mathcal J} t_\ell^*}{\min_{\ell'\in\mathcal I} t_{\ell'}^*}\Big)\Big],
\end{equation*}
proves our first convergence rate. 
As for equation \eqref{eq:bound}, note that 
\begin{equation*}
d_{\mathcal I}(f_1(\bone_n,\bone_L),\bone_n)=\frac{1}{\alpha}\,\ln\bigg(\max_{i,i'\in\mathcal I} \frac{\sum_{j=1}^n\sum_{\ell=1}^L\mathcal A_{ij\ell}}{\sum_{j=1}^n\sum_{\ell=1}^L\mathcal A_{i'j\ell}}\bigg)
\end{equation*}
and
\begin{equation*}
d_{\mathcal J}(f_2(\bone_n,\bone_L),\bone_L)=\frac{1}{\beta}\,\ln\bigg(\max_{\ell,\ell'\in\mathcal J} \frac{\sum_{i=1}^n\sum_{j=1}^n\mathcal A_{ij\ell}}{\sum_{i=1}^n\sum_{j=1}^n\mathcal A_{ij\ell'}}\bigg),
\end{equation*}
so that $\delta_{\bvec}\big((\xvec^{(1)},\tvec^{(1)}),(\xvec^{(0)},\tvec^{(0)})\big) \leq C$ with
\begin{equation*}
C= \rho\ln\bigg(\max_{i,i'\in\mathcal I} \frac{\sum_{j=1}^n\sum_{\ell=1}^L\mathcal A_{ij\ell}}{\sum_{j=1}^n\sum_{\ell=1}^L\mathcal A_{i'j\ell}}\bigg)+\frac{1}{\beta}\ln\bigg(\max_{\ell,\ell'\in\mathcal J} \frac{\sum_{i=1}^n\sum_{j=1}^n\mathcal A_{ij\ell}}{\sum_{i=1}^n\sum_{j=1}^n\mathcal A_{ij\ell'}}\bigg).
\end{equation*}
Solving $\frac{\rho^k}{1-\rho}C\leq \epsilon$ for $k$, we finally obtain
\begin{equation*}
 \norm{(\xvec^{(k)},\tvec^{(k)})-(\xvec^*,\tvec^*)}_{\infty} \leq \epsilon \qquad \forall k \geq \frac{\ln\big((1-\rho)\epsilon\big)-\ln(C)}{\ln(\rho)},
\end{equation*} 
which concludes the proof.
\end{proof}
\begin{rem}\label{rem:computation}
The relevance of Theorem \ref{cor:convergence} is twofold: it provides a convergence result for the power sequence \eqref{eq:power_sequence}  and an explicit bound on the number of iterations $k$ required to achieve a desired approximation accuracy.   In particular equation \eqref{eq:formula_k} implies that the higher the values of  $\alpha$ and $\beta$,  the smaller $k$. 
\end{rem}

\subsubsection*{An explanatory example} 
Before concluding this section, we want to show with a small example that there are situations where the eigenvector versatility is not well defined, whilst the $f$-eigenvector centrality is. 
Let us consider a multiplex $\mathcal G$ with four nodes $V_n = \{1,2,3,4\}$ and two layers $V_L = \{1,2\}$ and let us suppose that there is an undirected edge connecting nodes $1$ and $2$ 
on layer $\ell = 1$ and one undirected edge between nodes $3$ and $4$ on layer $\ell = 2$. A drawing of such network is shown on the right of Figure \ref{fig:example_connectedness}. 
It is easily understood that all the nodes are equally important, since they all play the same role in the multiplex.
Let $J_k\in\R^{2k\times 2k}$ and $P_k \in \R^{4k\times 4k}$ be the matrices 
\[J_k = \left[\begin{array}{cc}  & I_k \\ I_k & \end{array}\right], \quad  P_k = \begin{bmatrix} 
I_k & &  \\
 & J_k & \\
    & &  I_k
\end{bmatrix}.\]
The adjacency matrices of this multiplex network $\mathcal G$ thus are 
\[
A^{(1)} = \left[\begin{array}{cc} J_2 & 0 \\ 0 & 0 \end{array}\right]\quad  \text{ and } \quad 
A^{(2)} = \left[\begin{array}{cc} 0 & 0 \\  0 & J_2 \end{array}\right]
\]
where $0$ denotes here the $2\times 2$ matrix of all zeros. 
Let us now see what happens when we compute the eigenvector versatility of these nodes. 
Let $B = {\rm diag}(A^{(1)},A^{(2)}) + (\bone_2\bone_2^T - I_2)\otimes I_4$. Since $G_{\mathrm{agg}}$ is clearly not connected, $B$ is reducible by Proposition \ref{pro:versatility} and the eigenvector versatility is not unique. Indeed, we have that 
\[
P_2 B P_2 = B' = \left[
\begin{array}{cc}
B_1' & \\
 & J_4 B_1'J_4
\end{array}\right], \quad \text{where} \quad B_1' = \begin{bmatrix}J_2 & I_2  \\
I_2 &  0 \end{bmatrix}\, .
\]
Thus the matrices $B$ and $B'$ are similar through $P_2$ and their spectrum is fully determined by the spectra of the two $4\times 4$ blocks appearing on the 
main diagonal of $B'$. 
Moreover, the two matrices on the diagonal of $B'$ are themselves similar, and thus, clearly, the spectrum of $B$ will contain the same elements  as 
that of $B'_1$, each counted with multiplicity two. 
Consequently, the spectral radius of $B$ will be an eigenvalue 
with algebraic multiplicity two. 
Moreover, since $B'_1$ is the adjacency matrix of a path with four nodes, we know that its spectrum is
\[
\sigma(B'_1) = \{2\cos(\pi j/5)\}_{j=1}^4 \approx \{\pm 0.618, \pm 1.618\}\, .
\]
Thus, the matrix $B$ has four eigenvalues on the spectral circle. From the viewpoint of computing eigenvalues and eigenvectors, 
this is clearly an issue since there are two distinct eigenvectors associated with $\rho(B)=\rho(B_1')\approx 1.618$. 
Since $B' = P_2BP_2 = B_1' \oplus J_4 B_1'J_4$, the eigenspace of $B$ corresponding to its spectral radius is generated by 
\[
{\tt vec}(F_1)  = P_2\begin{bmatrix}\xvec_1 \\ \mathbf 0 \end{bmatrix}, \qquad {\tt vec}(F_2)  = P_2 \begin{bmatrix}\mathbf 0 \\ J_4 \xvec_1  \end{bmatrix}\, ,
\]
where $\mathbf 0$ is here the zero vector of length $4$ and 
\[
\xvec_1 = \sqrt{\frac 2 5}\begin{bmatrix}
\sin(\pi/5)\\
\sin(4\pi/5) \\
\sin(2j\pi/5) \\
\sin(3j\pi/5)
\end{bmatrix} \approx \begin{bmatrix} 0.3717 \\
 0.3717 \\
 0.6015 \\
 0.6015
\end{bmatrix}
\]
is the normalized Perron eigenvector of $B_1'$. 
This can lead to very different results in terms of the eigenvector versatility of the nodes. 
Indeed, it is equally probable to get:
\[{\tt eig\_ver}_1 = F_1\bone = \xvec_1 \quad \text{or} \quad {\tt eig\_ver}_2 = F_2\bone = J_4 \xvec_1 \]

Thus,  the eigenvector versatility may not only  be not well-defined, but also it might lead to centrality scores that 
are not meaningful, as in this case. 
Indeed, ${\tt eig\_ver}_1$ gives more importance to nodes $3$ and $4$, while ${\tt eig\_ver}_2$ 
identifies as more meaningful nodes $1$ and $2$, contradicting our intuition that all the nodes should have the same importance.

On the other hand, if $\cA$ is the adjacency tensor $\cA_{ij\ell}=A_{ij}^{(\ell)}$,  an easy computation reveals that $\sum_{ij}\cA_{ij\ell}=2$ and $\sum_{j\ell}\cA_{ij\ell}=1$, for any $\ell=1,2$ and any $i=1,\dots, 4$. Thus, independently of the choice of $\alpha$ and $\beta$, the pair  $(\mathbf 1, \mathbf 1)\in\R^4\times \R^2$ is a positive solution to \eqref{eq:model}. As $\cP_\cA^{\, 4\times 2}=\R^4_>\times \R^2_>$, Theorem \ref{thm:main_undirected} implies that  $(\mathbf 1, \mathbf 1)$ is the unique positive solution to \eqref{eq:model}  when $2/\beta<\alpha-1$  (up to positive multiples); therefore, the $f$-node eigenvector  
centrality of every node is $C_f(i)=1/4$, as one would intuitively expect. Moreover, the $f$-layer eigenvector centrality is $L_f(\ell)=1/2$, for any layer $\ell=1,2$ thus confirming the intuition that all the layers are equally important, since they are, as the nodes, interchangable.

\section{Numerical experiments}
\label{sec:ne}

In this section we will describe the results obtained when the $f$-eigenvector centrality is used to rank the nodes in a multiplex network. 
All experiments were performed using MATLAB Version
9.1.0.441655 (R2016b) on an HP EliteDesk running Scientific Linux 7.3 (Nitrogen), a 3.2 GHz Intel Core i7 processor, and 4 GB of RAM.
The experiments can be reproduced using the code available at {\tt https://github.com/ftudisco/node\_layer\_eigenvector\_centrality}.

We compare the results obtained when ranking the nodes according to the $f$-eigenvector centrality as in \dref{def:eig} 
with those obtained when using the eigenvector versatility {\tt eig\_ver}, the centrality measure {\tt eig\_cen} described in~\eqref{eq:EC}, 
the uniform eigenvector-like centrality {\tt agg\_eig},  and the aggregate 
degree centrality {\tt agg\_deg}. 
Recall that the aggregate degree of a node $i$ is the total number of links adjacent to it in the multiplex, i.e., $(A_{\rm agg}(\bone)\bone)_i$. 
Since we do not have any information regarding the importance of layers, nor the influence they have on each other, 
we are not testing the performance of the local and global heterogeneous eigenvector-like 
centrality measures. Indeed, as we have seen in Section~\ref{sec:related}, if no additional information is available, 
these reduce to either {\tt agg\_eig}, if $W = \bone\bone^T$, 
or {\tt eig\_cen}, if $W = I$. 

We want to stress once again that our iteration is able to return two centrality vectors: 
one for the importance of nodes and one for the importance of layers, while the other measures can only 
be computed for the nodes. 

Our model requires the selection of two positive 
scalars $\alpha$ and $\beta$ such that $2/\beta<(\alpha -1)$. 
Concerning the parameter $\beta$, a good choice could 
be to select it equal to $2$, since in the definition of $\tvec$ 
the importance of the layers is related to a quadratic polynomial of the node centralities. 
Therefore, in all our tests, we will use $\beta = 2$. 
Since there is no inferred way of selecting $\alpha>2$, we select it to be $\alpha = 2.1$ and 
we will study, later in this section, how 
the rankings change  when  we let its value vary.  
For the sake of completeness, we remark that we performed a similar experiment where we kept the value of $\alpha = 2 $ fixed while letting that of $\beta$ vary. We decided not to include the results here since 
they aligned with those obtained for fixed $\beta$ and varying~$\alpha$.

The sequence $\{(\xvec^{(k)},\tvec^{(k)})\}_k$ defined in Theorem~\ref{thm:main_undirected} and  used for the computation of the $f$-eigenvector centrality 
is an adaptation of the power method to handle $3^{\rm rd}$-order tensors. 
A normalization step 
is thus performed in order to avoid overflow and underflow in the computations. 
The stopping criterion for our algorithm will be the relative difference between two subsequent iterations. 
The algorithm stops when both the $f$-eigenvector centrality vectors have reached the desired level of accuracy. 
In more detail, our algorithm stops when 
$$
\max\left\{ \frac{\|\xvec^{(k)} - \xvec^{(k-1)} \|}{\|\xvec^{(k)}\|},   \frac{\|\tvec^{(k)} - \tvec^{(k-1)} \|}{\|\tvec^{(k)}\|} \right\} < {\rm tol}
$$


In the tests we used ${\rm tol} = 10^{-6}$ and the Euclidean norm $\|\xvec\|_2 = \sqrt{\sum_i x_i^2}$, since numerical tests not displayed here showed that 
smaller values of the tolerance and different vector norms returned the same rankings. 
Recall that we are not interested in the actual value of the centrality of each node and layer, but rather 
in determining the ranking position of each one of them.


We tested our technique on two multiplex networks with undirected and unweighted layers, which are studied separately in the following two subsections. 

\subsubsection*{European airlines  trasportation network}
The first dataset is the EU-air transportation multiplex~\cite{CGG13,DataDD} 
(\textsf{EUair}), which consists of 37
 layers, each of which represents a different European airline. 
The 450 nodes in each layer represent Europen airports and the links represent flights among them (see~\cite{DataDD} for a complete description of the datasets.) 
This multiplex does not have empty layers and there are 33 nodes with zero aggregate degree, which are correctly assigned a zero score by all the centrality measures. 
All the layers as well as the matrix $A_{\rm agg}(\bone)$ are disconnected, and hence all the centrality measures but $C_f$ and {\tt agg\_deg} are not uniquely determined.

Concerning the computation of $C_f$, we have that the layer centrality vector converges first after 21 iterations, while for the method to stop it took one more iteration.

\begin{table}
\centering
\caption{\textsf{EUair}: Pearson's correlation coefficient between any two centrality measures tested in this section. 
For the $f$-eigenvector centrality we used $\alpha = 2.1$ and $\beta = 2$.}
\label{tab:pears}
\begin{tabular}{c|ccccc}
& $C_f$          & {\tt eig\_ver} & {\tt eig\_cen} & {\tt agg\_eig} & {\tt agg\_deg} \\
\hline
$C_f$          &    -           &  0.89   & 0.81  &  0.86  &  0.88\\
{\tt eig\_ver} &   0.89         &    -    & 0.97  &  0.99  &  0.97\\
{\tt eig\_cen} &   0.81         &  0.97   &  -    &  0.98  &  0.94\\
{\tt agg\_eig} &   0.86         &  0.99   & 0.98  &   -    &  0.97\\
{\tt agg\_deg} &   0.88         &  0.97   & 0.94  &  0.97  &   -  \\
\hline
\end{tabular}
\end{table}

\begin{figure}
\includegraphics[width=.9\textwidth]{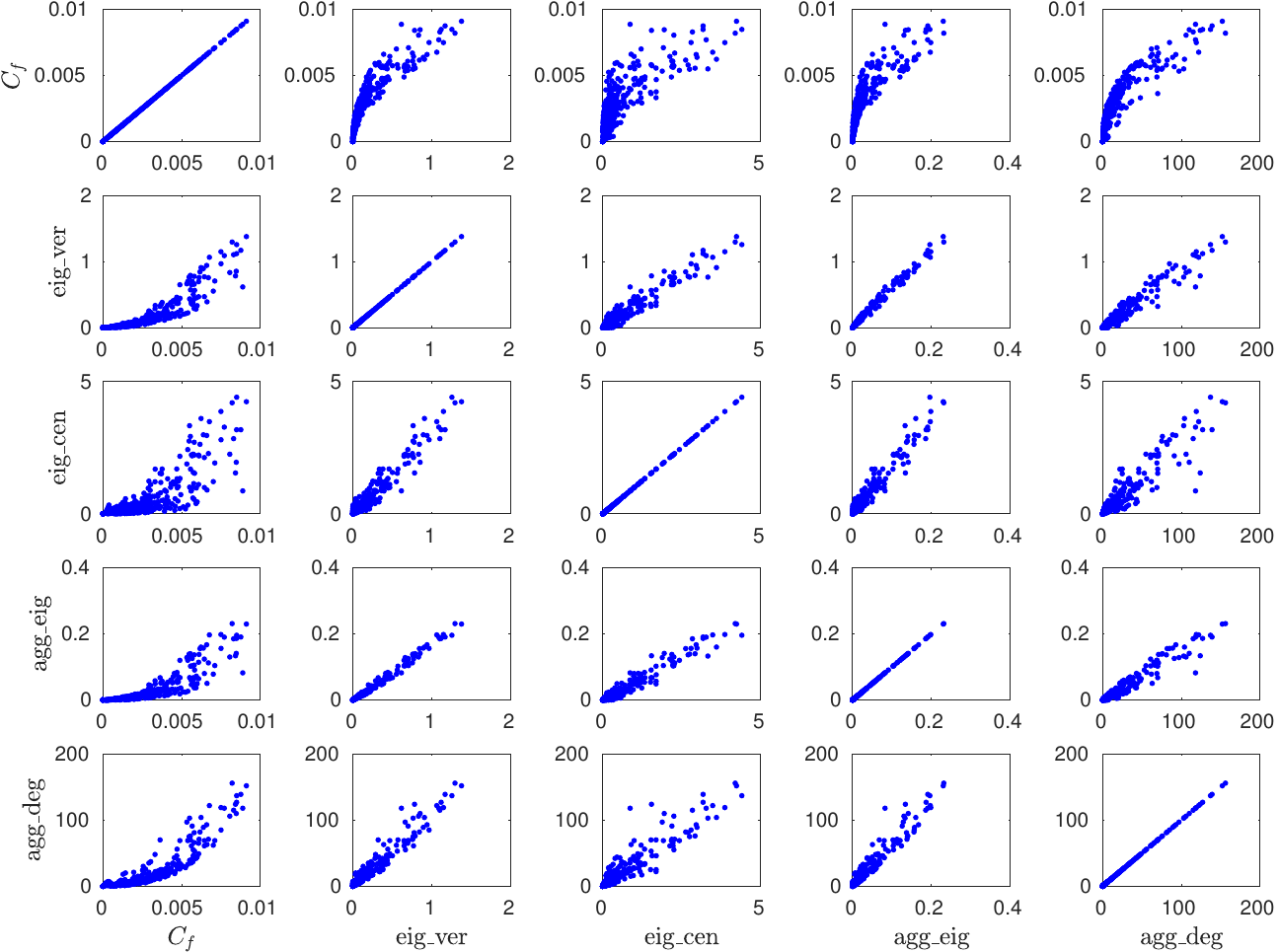}
\caption{\textsf{Euair}: Scatter plots between the different centrality vectors. For the $f$-eigenvector centrality we used $\alpha = 2.1$, $\beta = 2$.}
\label{fig:scat_EUair}
\end{figure}

Table~\ref{tab:pears} contains the Pearson's correlation coefficient between any two of the centrality vectors computed 
in the test. 
This coefficient is a measure of linear correlation and varies between $-1$ and $1$, $1$ indicating perfect linear correlation. 
From the table, we see that every pair of centrality vectors displays a high value of this coefficient, with a peak at 
$0.99$ between {\tt eig\_ver} and {\tt agg\_eig}. 
The lower values for this coefficient of linear correlation are achieved when comparing $C_f$ with 
other vectors. 
It is easily understood, however, that correlation coefficients do not represent the best way to compare 
ranking vectors; indeed, two vectors may display a high linear correlation, but the rankings provided by them might   
significantly differ. 
To have a better understanding of our results, we thus show in Figure~\ref{fig:scat_EUair} the scatter plot between the different 
centrality vectors. 
It is straightforward to see that, although $C_f$ has a high correlation with the other measures, the derived rankings are pretty 
different. 
Indeed, the nodes having the highest score with respect to the $f$-eigenvector centrality do not correspond to those ranked at the top 
by the other measures. This can also be seen from Table~\ref{tab:cen_EUair}, where we have listed the top 10 ranked nodes according to the different centrality measures.  
Although the first ranked node according to the $f$-eigenvector centrality is ranked first by the eigenvector versatility and second by all the other measures,  
when it comes to the second ranked node according to $C_f$,  
we find it to be ranked, e.g., $34$th by the eigenvector versatility and $71$st by {\tt eig\_cen}.  

Figure \ref{fig:europa} shows the geographical locations of the top five European airports according to the computed $f$-node eigenvector centrality (left) and eigenvector versatility (right).

\begin{table}
\centering
\caption{\textsf{EUair}: top 10 ranked nodes according to the centrality measures tested in this section. 
For the $f$-eigenvector centrality we used $\alpha = 2.1$ and $\beta = 2$.}
\label{tab:cen_EUair}
\begin{tabular}{c|cccccccccc}
\hline
$C_f$          &   50  &  12 &   38  &  40  &   2 &  108  & 252  & 166  &  15 &   57 \\
{\tt eig\_ver} &   50  &  15 &   40  &  38  &  83 &    2  & 166  &   7  &  64 &   34 \\
{\tt eig\_cen} &   40  &  50 &   15  &  83  &  22 &   64  &  14  &   7  &  38 &    2 \\
{\tt agg\_eig} &   15  &  50 &   83  &  64  &  40 &   38  &   7  &   2  & 166 &   66 \\
{\tt agg\_deg} &   15  &  50 &   38  &  40  &   2 &  252  &  64  &  83  &   7 &   12 \\
\hline
\end{tabular}
\end{table}

\begin{figure}
    \begin{tikzpicture}[scale=.89]
    \node[anchor=south west,inner sep=0] at (0,0) {\includegraphics[width=.408\textwidth,clip,trim=2cm 0 2cm 0]{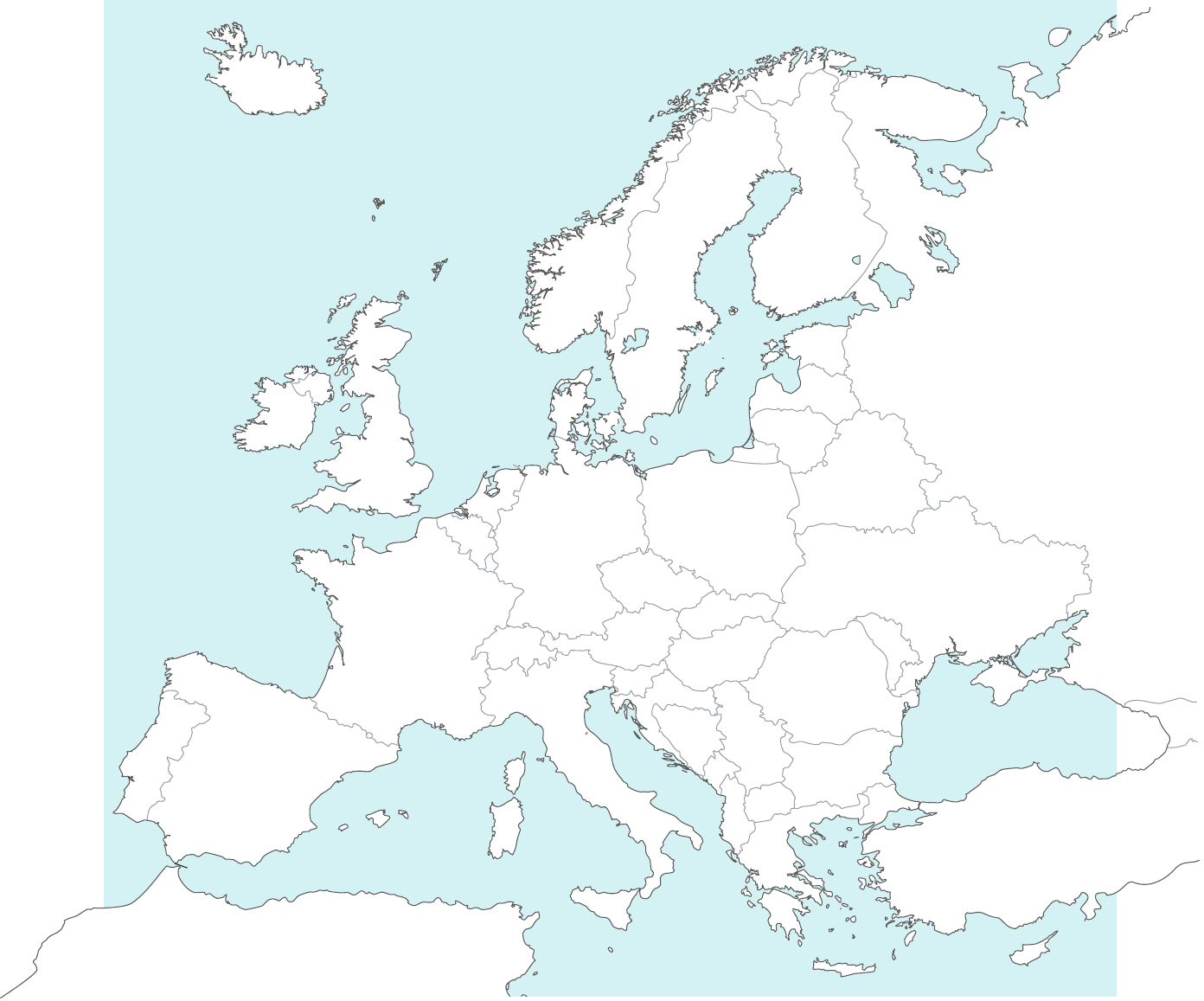}};
    \draw[gray,thick] (0,0) rectangle (6,6);
    \draw[draw=black, fill=red, opacity=0.7] (.85,1.4) circle (.3);
    \draw[draw=black, fill=red, opacity=0.7] (1.8,3.1) circle (.25);
    \draw[draw=black, fill=red, opacity=0.7] (2.9,2.35) circle (.2);
    \draw[draw=black, fill=red, opacity=0.7] (1.6,1.43) circle (.15);
    \draw[draw=black, fill=red, opacity=0.7] (2.55,2.65) circle (.1);
    \end{tikzpicture}
    \hfill
    \begin{tikzpicture}[scale=.89]
    \node[anchor=south west,inner sep=0] at (0,0) {\includegraphics[width=.408\textwidth,clip,trim=2cm 0 2cm 0]{europe2.eps}};
    \draw[gray,thick] (0,0) rectangle (6,6);
    \draw[draw=black, fill=red, opacity=0.7] (.85,1.4) circle (.3);
    \draw[draw=black, fill=red, opacity=0.7] (2.35,3.0) circle (.25);
    \draw[draw=black, fill=red, opacity=0.7] (1.6,1.43) circle (.2);
    \draw[draw=black, fill=red, opacity=0.7] (2.9,2.35) circle (.15);
    \draw[draw=black, fill=red, opacity=0.7] (2.93,1.28) circle (.1);
    \end{tikzpicture}
    \caption{Red dots show the geographical locations of the top five European airports according to $C_f$ (left) and {\tt eig\_ver} (right). The larger the dot, the higher the corresponding ranking.}\label{fig:europa}
\end{figure}

To provide further insight and have a better understanding of how far the $f$-node eigenvector centrality is 
from the other centrality measures, we look at the {\it intersection similarity}~\cite{FKS03} 
between the derived ranking vectors. 
The intersection similarity is a measure used to compare the top $K$ entries of two ranked lists that 
may not contain the same elements.  It is defined as follows: let $\calL^{1}$ and $\calL^{2}$ be two 
ranked lists, and let us call $\calL^{j}_k$ the list of the top $k$ elements listed in $\calL^{j}$, for 
$j=1,2$. 
Then, the {\it top $K$ intersection similarity between $\calL^{1}$ and $\calL^{2}$} is defned as 
\begin{equation*}
{\rm isim}_K(\calL^{1},\calL^{2}) = \frac{1}{K}\sum_{k=1}^K\frac{|\calL^{1}_k\Delta\calL^{2}_k|}{2k},
\label{eq:isim}
\end{equation*}
where $|S|$ denotes the cardinality of the set $S$ and $\calL^{1}_k\Delta\calL^{2}_k$ is the symmetric difference between $\calL^{1}_k$ and $\calL^{2}_k$. 
When the ordered sequences contained in $\calL^{1}$ and $\calL^{2}$ are completely different, then the intersection similarity between the two is maximum and it is equal to 1. 
On the other hand, the intersection similarity between two lists is equal to $0$ 
if and only if the two ordered sequences coincide.   

\begin{figure}
\includegraphics[width=\textwidth]{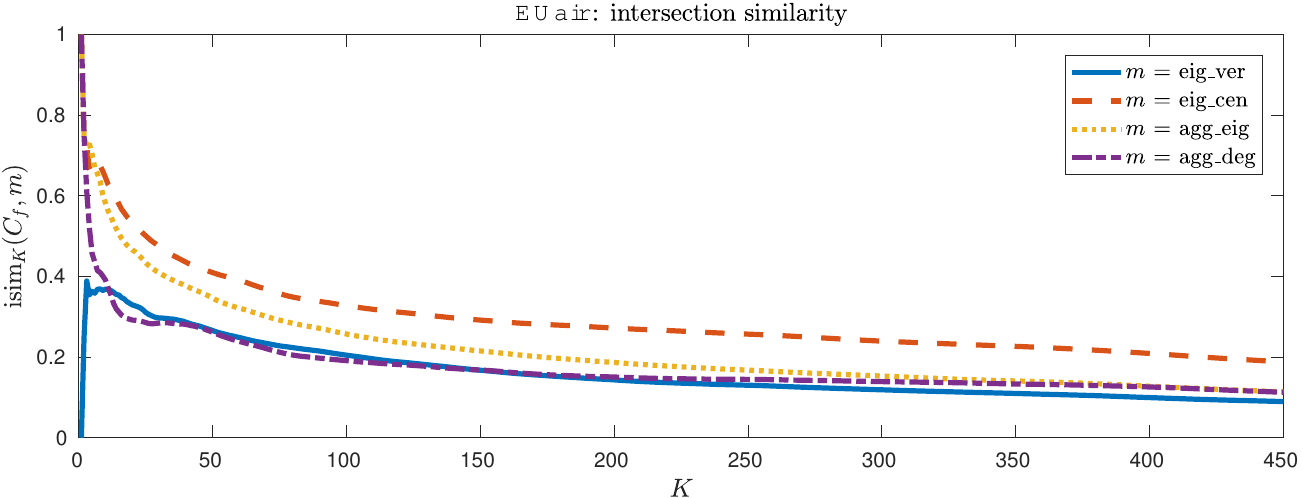}
\caption{\textsf{EUair}: plot of the intersection similarity ${\rm isim}_K(C_f,\cdot)$ between the ranking derived from the $f$-eigenvector centrality and other ranking 
vectors, for $K = 1,2,\ldots,n$. For the computation of $C_f$ we used $\alpha = 2.1$ and $\beta = 2$. Here $C_f$, {\tt eig\_ver}, {\tt eig\_cen}, {\tt agg\_eig},  and {\tt agg\_deg} denote the associated ranking vectors.}
\label{fig:isim}
\end{figure}

In Figure~\ref{fig:isim} we display the evolution of the intersection similarity, as a function of $K$, between the ranking derived from $C_f$ and 
those obtained using the scores provided by the other four centrality measures. 
Figure~\ref{fig:isim} shows that, especially for small values of $K$, the intersection similarity between the ranking derived from $C_f$ and 
all the other centralities is high. The only exception is represented by the eigenvector versatility; however, its  intersection similarity with $C_f$ grows quickly and reaches high values (around 0.4) already for small values of $K$ (say, $K=5$). 
This tells us that the list of the top ranked nodes obtained with the new centrality measure differs significantly from the other measures, as the results in 
Table~\ref{tab:cen_EUair} already showed for $K = 10$. 

Before moving on to the second data set, we want to understand how different choices of the parameters influence the derived rankings. 
 In Figure~\ref{fig:varya_EUair} we display the results obtained when $\beta=2$ is kept fixed and $\alpha = [2.1, 2.5, 2.7, 3, 4, 5, 10]$: the top plots are spaghetti plots representing the changes in the rankings, while 
the bottom plots display the time (in seconds) and the number of iterations required for the computation of the $f$-eigenvector centrality 
for the different choices of $\alpha$. 
It is clear from the top plots that the rankings do not depend heavily on the choice of $\alpha$, since the lines in the spaghetti plot (each representing the ranking of a certain node) are almost all completely horizontal. 
Concerning the bottom plots, both iterations count and execution time decrease as $\alpha$ increases, confirming that the convergence rate follows the powers of $\rho=\frac{\sqrt{8\alpha+\beta}+\sqrt{\beta}}{2\alpha\sqrt{\beta}}$,  as discussed in Theorem \ref{cor:convergence}. 
In more detail, since $\beta = 2$, we have that $\rho\approx 1/\sqrt{\alpha}$; thus, from \eqref{eq:formula_k} it follows that the number of iterations $k_*$ required to achieve convergence decays as the inverse of the logarithm of $\alpha$, i.e., there exist constants $c_1$ and $c_2$ such that  $k_* \approx \phi(\alpha)=c_1/\ln(\alpha)+c_2$.
We display $\phi(\alpha)$ for suitable values of $c_1$ and $c_2$ with a solid line in the bottom right plot of Figure~\ref{fig:varya_EUair}. The actual number of iterations required in our computations for different values of 
$\alpha$ closely follows the expected behaviour described by $\phi(\alpha)$. In particular, higher values of $\alpha$ correspond to fewer iterative steps, as already noticed in Remark~\ref{rem:computation}. 
This has computational relevance; indeed, it is not difficult to observe that when $\mathcal G$ has sparse layers all the measures can be computed with $O(Ln)$ flops per iterative step. However, Figure~\ref{fig:varya_EUair} shows that tuning $\alpha$  allows us to reduce the overall timing and number of iterations required to compute our centrality vector, while simultaneously returning a meaningful ranking.

\begin{figure}
\includegraphics[width=.9\textwidth]{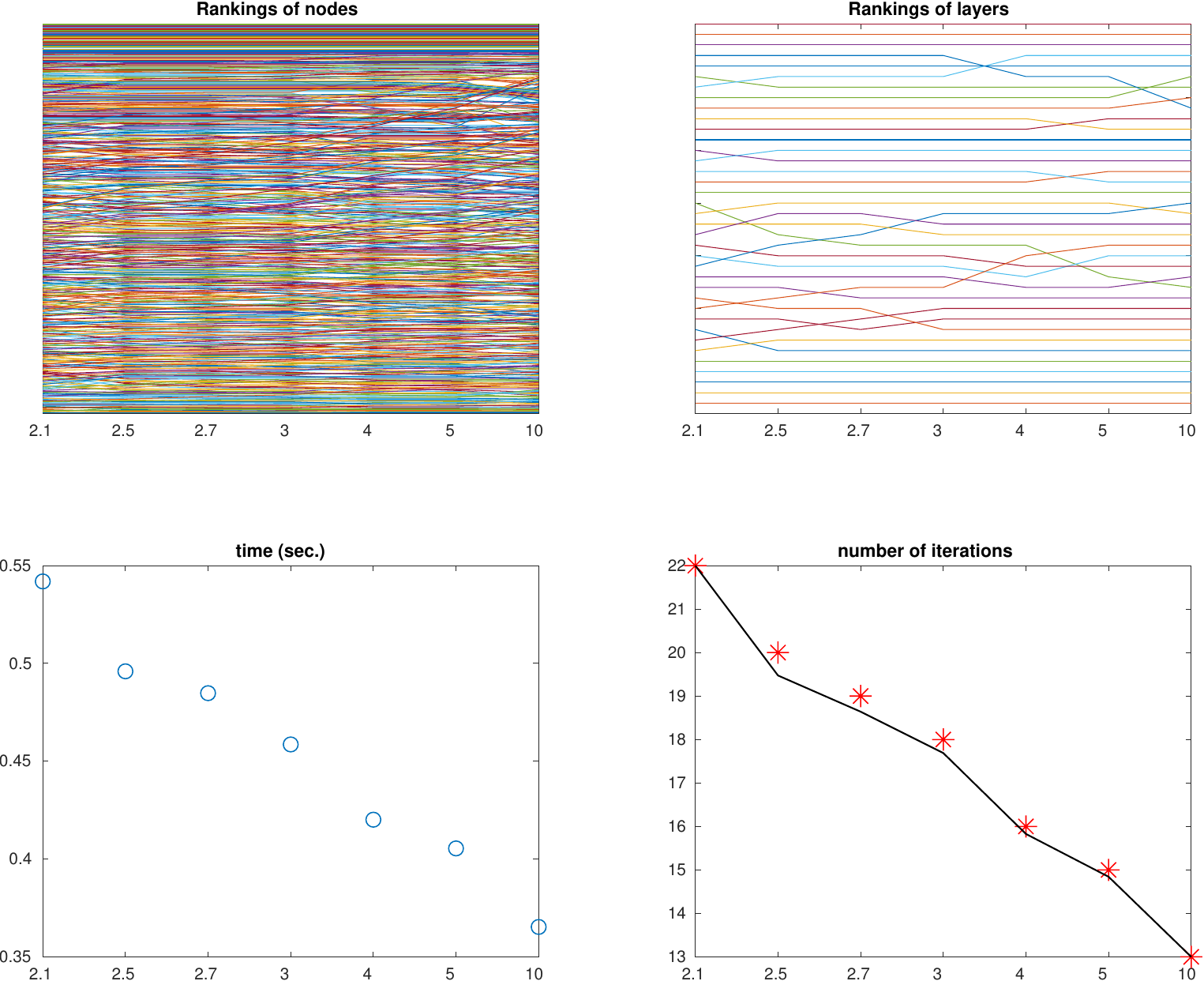}
\caption{\textsf{Euair}: Top - evolution of the rankings obtained using the $f$-eigenvector centrality for nodes and layers when $\alpha = [2.1, 2.5, 2.7, 3, 4, 5, 10]$ and $\beta = 2$. Bottom - timing in seconds and overall number of iterations required to compute $C_f$ when $\alpha$ varies, with the solid line representing the curve $c_1/\ln(\alpha)+c_2$ for suitable values of $c_1$ and $c_2$.}
\label{fig:varya_EUair}
\end{figure}

\begin{figure}
\includegraphics[width=.9\textwidth]{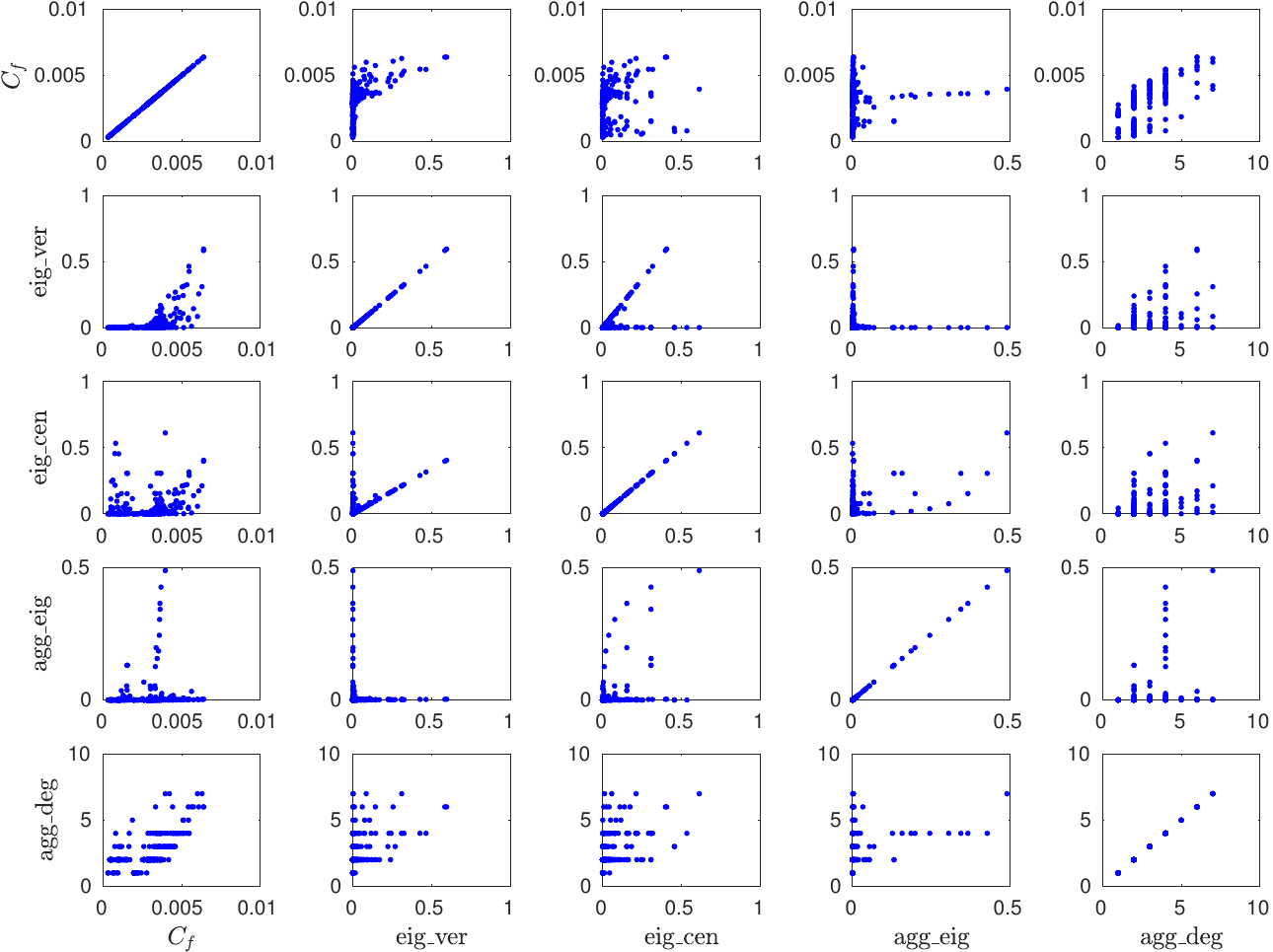}
\caption{\textsf{London}: Scatter plots between the different centrality vectors. $\alpha = 2.1$, $\beta = 2$.}
\label{fig:scat_london}
\end{figure}

 \subsubsection*{London city trasportation network}
We now move on to the dataset \textsf{London}, which represent the London city transportation network~\cite{DataDD,DDSBA14}. 
The 369 nodes correspond to train stations and the existing routes between them are the edges in each of the three layers this multilayer consists of.
The stations of the Underground, Overground, and DLR are considered.
We used the unweighted version of this undirected multiplex network. The matrix $A_{\rm agg}(\bone)$ is irreducible and hence both {\tt eig\_ver} and {\tt agg\_eig} are well-defined for this dataset. 
All layers are non-empty and disconnected, and the nodes have all positive aggregate degree. 
Therefore the $f$-eigenvector centrality vectors will be positive. 
As before, we pick $\alpha = 2.1$ and $\beta =2$.

The layer centrality vector converges first (23 iterations), while the other vector converges at iteration 24, thus making the method stop. 
The Pearson's correlation coefficients between any two centrality vectors are displayed in Table~\ref{tab:pears2}. 
For this example, no pair of vectors displays linear correlation. 
This can be seen also by looking at the scatter plots in Figure~\ref{fig:scat_london}.  

\begin{table}
\centering
\caption{\textsf{London}: Pearson correlation coefficient between any two centrality measures tested in this section. 
For the $f$-eigenvector centrality we used $\alpha = 2.1$ and $\beta = 2$.}
\label{tab:pears2}
\begin{tabular}{c|ccccc}
& $C_f$          & {\tt eig\_ver} & {\tt eig\_cen} & {\tt agg\_eig} & {\tt agg\_deg} \\
\hline
$C_f$          &  -      &   0.55  &  0.20 &   0.10  &  0.60\\
{\tt eig\_ver} &    0.55 &   -     &  0.53 &  -0.06  &  0.44\\
{\tt eig\_cen} &    0.20 &   0.53  &   -   &   0.44  &  0.48\\
{\tt agg\_eig} &    0.10 &  -0.06  &  0.44 &     -   &  0.31\\
{\tt agg\_deg} &    0.60 &   0.44  &  0.48 &   0.31  &  -   \\
\hline
\end{tabular}
\end{table}

\begin{figure}
\includegraphics[width=\textwidth]{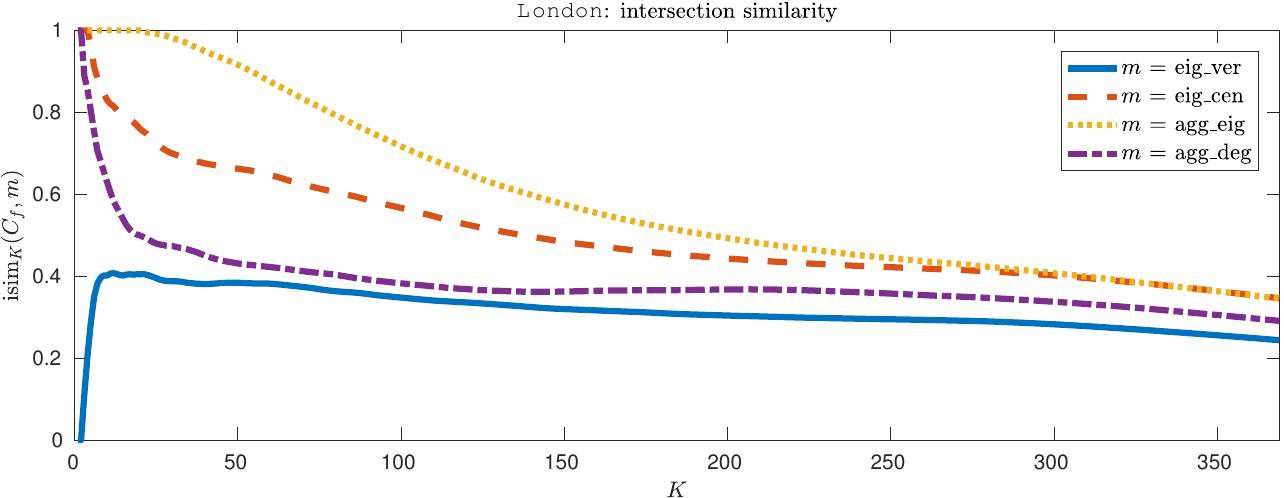}
\caption{\textsf{London}: plot of the intersection similarity ${\rm isim}_K(C_f,\cdot)$ between the ranking derived from the $f$-eigenvector centrality and other ranking 
vectors, for $K = 1,2,\ldots,n$. For the computation of $C_f$ we used $\alpha = 2.1$ and $\beta = 2$. Here $C_f$, {\tt eig\_ver}, {\tt eig\_cen}, {\tt agg\_eig},  and {\tt agg\_deg} denote the associated ranking vectors.}
\label{fig:isim2}
\end{figure}

In Figure~\ref{fig:isim2} we display the intersection similarity between the ranking obtained from $C_f$ and the other ranking vectors as a function of $K$, i.e., the number of
top ranked elements considered. 
As for the previous dataset, the value of the intersection similarity between the top $K$ entries of the ranking vector derived from $C_f$ and all the 
other methods is large; this shows that, especially for small $K$, the top ranked nodes according to the $f$-eigenvector centrality differ from those 
identified as most important from the other networks. 
To show this, in Table~\ref{tab:cen_london} we list the top 10 ranked nodes according to the different centrality measures. 

\begin{table}
\centering
\caption{\textsf{London}: top 10 ranked nodes according to the centrality measures tested in this section. 
For the $f$-eigenvector centrality we used $\alpha = 2.1$ and $\beta = 2$.}
\label{tab:cen_london}
\begin{tabular}{c|cccccccccc}
\hline
$C_f$          &    69  &  68 &   28 &  181 &  182  &  35 &   46 &   29 &  214 &    9 \\
{\tt eig\_ver} &    69  &  68 &  214 &   29 &  215  & 207 &   28 &  282 &  121 &  181 \\
{\tt eig\_cen} &     4  &  13 &  291 &  325 &   69  &  68 &  214 &  339 &  261 &  225 \\
{\tt agg\_eig} &     4  & 225 &  226 &  259 &  306  & 305 &  260 &  264 &  339 &    3 \\
{\tt agg\_deg} &     4  &  28 &  182 &  220 &    2  &  35 &   46 &   50 &   68 &   69 \\
\hline
\end{tabular}
\end{table}

\section{Conclusions and future work}
In this paper we have introduced the $f$-eigenvector centrality:  a new multi-dimensional eigenvector-based centrality measure for nodes and layers in multiplex networks with undirected layers. 
We have shown that in order to guarantee well-posedness of the definition, non-linearity has to be introduced in the model extending the classical Bonacich index to the multidimensional setting of multiplexes. 
We have further proved that existence and uniqueness of the $f$-eigenvector centrality can be guaranteed for any non-negative $3^{\rm rd}$-order tensor that satisfies very mild conditions. In particular the $f$-eigenvector centrality can be computed efficiently without any a priori analysis of the irreducibility of the multi-layer network. 
We compared the newly introduced centrality measures with the eigenvector-based centrality measures found in the literature, and we showed that it provides different rankings on the two real world data sets we tested. 

Future work will focus on  the extension of these results to the case of multiplex networks with directed layers. We also plan to investigate how to introduce in our model any available information on the importance of layers that is not deducible from the network topology. 

\section{Acknowledgements}
This work used pre-existing data that is publicly available from {\tt http://deim.urv.cat/$\sim$manlio.dedomenico/data.php}.

The authors would like to thank two anonymous referees for their comments.


\end{document}